\documentclass[journal,draftcls,onecolumn,12pt,twoside]{IEEEtranTCOM}
\normalsize

\usepackage[noadjust]{cite}

\ifCLASSINFOpdf
\else
\fi

\usepackage{algorithm}
\usepackage{algpseudocode}
\usepackage{textcomp}
\usepackage{pgfplots}
\usepackage{float}
\usepackage[compact]{titlesec}
\usepackage{relsize}
\usepackage{booktabs, longtable, multirow, supertabular, tabularx}

    \newcolumntype{P}[1]{>{\centering\hspace{0pt}\arraybackslash}p{#1}}
    \newcolumntype{M}[1]{>{\centering\hspace{0pt}\arraybackslash}m{#1}}
    \newcolumntype{L}{>{\centering\arraybackslash}m{3cm}}

\usepackage{cite}
\def\BibTeX{{\rm B\kern-.05em{\sc i\kern-.025em b}\kern-.08em
    T\kern-.1667em\lower.7ex\hbox{E}\kern-.125emX}}
\usepackage{textcomp}
\let\labelindent\relax
\usepackage{enumitem}
\setlist{nolistsep,leftmargin=.6cm}
\usepackage{bbm}
\usepackage{amsmath, amssymb, amsthm}
\usepackage{psfrag}
\usepackage{graphicx}
\usepackage{tikz}
\usetikzlibrary{arrows.meta,
                fit,
                positioning,
                quotes}

\usepackage{cuted}
\usepackage[font=small]{caption}
\usepackage{pbox}
\usepackage{filecontents}
 \newcommand{\ind}{\perp\!\!\!\!\perp}
\begin{document}
\newtheorem{proposition}{Proposition}
\newtheorem{definition}{Definition}
\newtheorem{lemma}{Lemma}
\newtheorem*{theorem*}{Theorem}
\newtheorem{theorem}{Theorem}
\newtheorem{corollary}{Corollary}
\newtheorem{assumption}{Assumption}
\newtheorem{claim}{Claim}
%
\title{Non-Coherent Active Device Identification for Massive Random Access }
%
%
%

\author{Jyotish~Robin,~\IEEEmembership{Graduate~Student~Member,~IEEE,}
        Elza~Erkip,~\IEEEmembership{Fellow,~IEEE}
\thanks{Jyotish Robin and Elza Erkip are with the Department of Electrical and
Computer Engineering, NYU Tandon School of Engineering, 6 MetroTech Center, Brooklyn, NY 11201 USA. (e-mail: jyotish.robin@nyu.edu). }
}

%
%

\markboth{IEEE Transactions on Communications}%
{Submitted paper}
%



\maketitle

\vspace{-1cm}
\begin{abstract}
Massive Machine-Type Communications (mMTC) is a key service category in the current generation of wireless networks featuring an extremely high density of energy and resource-limited devices with sparse and sporadic activity patterns.   In order to enable random access in such mMTC networks, base station needs to identify the active devices while operating  within stringent access delay constraints.  In this paper, an energy efficient active device identification protocol is proposed in which active devices transmit  On-Off Keying (OOK) modulated preambles jointly and  base station employs non-coherent energy detection avoiding channel estimation overheads. The minimum number of channel-uses required by the active user identification protocol is characterized in the asymptotic regime of total number of devices $\ell$ when the number of active devices $k$ scales as $k=\Theta(1)$  along with an achievability scheme relying on the equivalence of activity detection to a group testing problem.  Several practical schemes based on  Belief Propagation (BP) and  Combinatorial Orthogonal Matching Pursuit (COMP) are also proposed. Simulation results show that BP strategies outperform COMP significantly and can operate close to the theoretical achievability bounds. In a partial-recovery setting where few misdetections are allowed, BP continues to perform well.

\end{abstract}


%
\IEEEpeerreviewmaketitle

\section{Introduction}
%
%
%
%
\IEEEPARstart{S}{ixth} generation (6G) wireless communication networks are expected to serve a myriad of smart sensing applications involving massive machine-type communications (mMTC) as in Internet of Things (IoT). In contrast to the traditional cellular systems, mMTC has to tackle  a novel set of technical challenges posed by the diversity in  data traffic patterns and application-specific constraints in machine-centric communications \cite{7736615}. Typically,  mMTC networks are  comprised of sensors, tracking devices, actuators, etc., which engage in the monitoring and recording of critical information. These  mMTC devices often need to support dynamic service  requirements such as variable payload sizes, data security and delay constraints while reliably operating on low duty cycles and stringent  energy budget constraints to enhance efficiency \cite{7263368}.   Machine-type applications naturally lead to a sparse and sporadic data traffic since only a small subset of devices are active at any given time instant \cite{9023459}. Moreover, since the payload size is small and data transmissions are infrequent, On-Off Keying (OOK) based modulation schemes are known to be energy efficient as the devices can save energy during off-slots \cite{5937958,8534548,9596585}. Furthermore, in the context of green communication, there has been a renewed interest in OOK modulation schemes for serving low-rate and high-reliability  applications \cite{6120373}.

Another key distinction of mMTC systems w.r.t. cellular systems  is that they are more prone to security threats since they are composed of simple nodes with limited computational capabilities \cite{8386824}. Hence, it is critical to maintain data security and provide robustness against malicious attacks such as jamming, spoofing and denial-of-service \cite{mendez2018internet}. To combat malicious jamming attacks, spread-spectrum techniques, viz., direct
sequence spread spectrum  and frequency hopping
spread spectrum (FHSS), have been considered for commercial
wireless networks \cite{7036828}. Specifically, FHSS approaches are more suitable in IoT scenarios due to their scalability and superior resilience to narrow band interference  \cite{fi11010016}. For example,  many mMTC architectures rely on fast frequency hopping where each symbol experiences independent fading to enhance communication security against jamming attacks \cite{4394775, 9136922}.  

In mMTC networks, the BS faces the challenging task of successfully decoding data packets  transmitted by a random unknown active subset of sensors. In fact,  efficiently identifying the active devices is a critical part of many Random Access (RA) architectures \cite{8403656,TS36.213,9266124} as misidentifications can lead to  data losses which  significantly impact critical IoT environments \cite{ALTURJMAN2017299}. In  this paper, our focus is on this key problem of active device identification in  mMTC networks  which are operating under stringent requirements in terms of latency, data rates, reliability, energy consumption, and
security/privacy. In dense mMTC networks, traditional grant-based random access strategies result in excessive preamble collisions drastically increasing the signaling overheads and transmission delays \cite{9060999}. On the other hand, grant-free random access suffers from severe co-channel interference caused by the non-orthogonality of preambles \cite{9537931}. Thus, active device identification is a more intricate problem in massive mMTC systems than traditional cellular systems and requires novel strategies.

 Several recent works \cite{8734871,7952810,8454392,5695122,7282735}  exploit the sporadic device activity pattern to employ compressed sensing techniques such as Approximate Message Passing (AMP) \cite{8734871,7952810,8454392,5695122} and sparse graph codes \cite{7282735} to detect active users. In AMP based approaches, preamble sequences are typically generated from i.i.d. Gaussian distributions and the activity  detection performance in massive user settings can be characterized using state evolution  \cite{5695122}. Senel \textit{et al.} in \cite{8444464} proposed an activity detection scheme by using symmetric Bernoulli pilot sequences under the assumption of state evolution. In \cite{8262800}, an RA protocol for short packet transmission was presented by viewing activity detection as a  group testing (GT) problem. Our previous works in \cite{9500808,9593173} address the activity detection problem  taking into account the  access delay and energy constraints.  Chen \textit{et al}. in \cite{7852531,6691257} proposed a novel many-access channel (MnAC) model for systems with many transmitters and a single receiver in which the number of transmitters grows with the blocklength.  Another related work is \cite{8849288} where the problem of designing unsourced  massive random  access
architectures for an AWGN  random access channel with random fading
gains unknown to the receiver is discussed. 

Most of the existing literature assumes either a Gaussian channel or a block-fading channel model where the channel remains static in each block in which channel state information (CSI) can be accurately obtained at the BS through the use of pilots before activity detection. For secure mMTC networks such as the FHSS systems mentioned before, a more accurate channel model would incorporate fast fading where reliable channel estimates cannot be obtained \cite{923716}. 

In  this paper, unlike the coherent detection approaches, we focus on efficiently detecting the active devices in an mMTC system using a threshold-based energy detection which is  a low-complex energy efficient  non-coherent method \cite{5174497,7514754}. Active devices in the network transmit their OOK modulated unique preambles synchronously and  BS performs energy detection. Our aim is estimate the set of active devices reliably using the  binary energy detector outputs corresponding to the joint OOK preamble transmissions without CSI knowledge. Note that our framework can accommodate a massive number of users since we are not constrained to orthogonal preambles. In addition, as preambles are designed independently for each user, new users can be added without network-wide changes thereby promoting  scalability.

The above framework for recovering the sparse set of active devices using binary measurements can be viewed as an instance of the well known GT  problem 
\cite{8926588}. In essence,  GT aims to identify a small number of ``defective''
items within a large population by performing tests on pools of items. A test
is classified as positive if the pool contains at least one defective, and negative if it contains
no defectives. Though connection between GT and RA  has previously been exploited in the literature \cite{1096146,8945,1057026,8849823,article}, the channel models used are often too simplistic or  unrealistic for many secure mMTC networks such as the one using FHSS techniques to ensure data security. In our paper, we bridge this gap by incorporating a fast fading channel to the GT model. To the best of our knowledge, this is the first work that integrates an OOK modulated preamble signaling at the sensors, a fast fading channel model and a non-coherent threshold-based energy detection at the BS for identifying the sparsely active devices in an mMTC network.
 
The main contributions of this paper can be summarized as follows:
\begin{itemize}
\item 
We propose a novel non-coherent activity detection strategy  for secure energy-constrained mMTC networks in which active devices jointly transmit their OOK-modulated preambles  and  the BS uses a non-coherent energy detection strategy to identify the active devices.

\item We use  information theoretic  tools to characterize the minimum user identification cost,  defined as the minimum number of channel-uses required to identify the active users \cite{9517965,7852531}.  Using results from GT \cite{8926588}, we derive an achievable minimum user identification cost for our activity detection framework in the $k=\Theta(1)$ regime as well as a valid lower bound in the general sub-linear regime where $k=\Theta(\ell^\alpha); 0 \leq \alpha<1$ which is shown to be tight when  $\alpha =0$. Here, $k$ denotes the number of active devices whereas $\ell$ denotes the total number of devices.

\item \sloppy We present and study several practical schemes  for active device identification in the non-coherent OOK setting based on Noisy-Combinatorial Orthogonal Matching Pursuit (N-COMP) and Belief Propagation (BP) techniques \cite{6120391,5169989}. We also present the  performance gap incurred by these practical strategies in comparison to our theoretical characterization of minimum user identification cost. We show that BP can come very close, and  with some practical relaxations such as partial recovery where few misdetections and false positives are allowed, can even exceed  the theoretical bounds.

   \end{itemize}

The remainder of this work is structured as follows. In Section II, we define the system model for the problem of active device identification in an mMTC network and establish its equivalence to a GT problem. In addition, we formally define the key metric of interest, viz, the minimum user identification cost. In Section III, we derive the minimum user identification cost for the non-coherent many-access channel (MnAC) with $\ell$ total and $k$ active devices by viewing active device identification as a decoding problem in an equivalent point-to-point communication channel. Moreover, we present a decoder achieving the minimum user identification cost in the $k =\Theta(1)$ regime which has severe computational complexity. In Section IV, we present several practical strategies for  user identification  under the constraints of   exact recovery and partial recovery and compare them against the minimum user identification cost. Section V concludes the paper.


\section{System model}~\label{sec:sysmodl} 
\subsection{Network model}
Consider an mMTC network consisting of $\ell$ users belonging to the set $\mathcal{D}=\{1,2,\ldots,\ell\}$. Among the $\ell$ users, only a subset of $k$  are active and wants to access the channel. We assume that the value of $k$ is periodically estimated at the BS and hence known \textit{a-priori} \cite{4085381}. The active user set  denoted by $\mathcal{A}=\{a_1,a_2\ldots, a_k\}$  is assumed to be randomly chosen among all ${\ell \choose k} $ subsets of size $k$ from $\mathcal{D}$ and is unknown to the BS. The BS aims to identify the active user set $\mathcal{A}$ in a time-efficient manner so that they could be allocated required resources for  prompt channel access. 

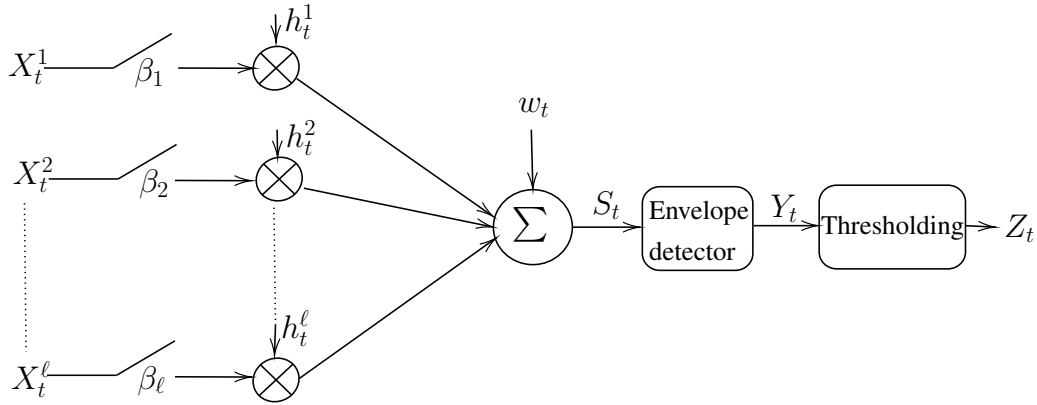
\begin{figure}   
	\centering
	\resizebox{5.5in}{!}{
\tikzset{every picture/.style={line width=1pt}} 

\begin{tikzpicture}[x=1pt,y=0.95pt,yscale=-1,xscale=1]

\draw  [dash pattern={on 0.84pt off 2.51pt}]  (55.67,152.17) -- (56.08,193.14) -- (56.67,251.17) ;
\draw    (153,62) -- (199.67,62.16) ;
\draw [shift={(201.67,62.17)}, rotate = 180.2] [color={rgb, 255:red, 0; green, 0; blue, 0 }  ][line width=0.75]    (10.93,-3.29) .. controls (6.95,-1.4) and (3.31,-0.3) .. (0,0) .. controls (3.31,0.3) and (6.95,1.4) .. (10.93,3.29)   ;
\draw    (151,137) -- (197.67,137.16) ;
\draw [shift={(199.67,137.17)}, rotate = 180.2] [color={rgb, 255:red, 0; green, 0; blue, 0 }  ][line width=0.75]    (10.93,-3.29) .. controls (6.95,-1.4) and (3.31,-0.3) .. (0,0) .. controls (3.31,0.3) and (6.95,1.4) .. (10.93,3.29)   ;
\draw    (151,268) -- (197.67,268.16) ;
\draw [shift={(199.67,268.17)}, rotate = 180.2] [color={rgb, 255:red, 0; green, 0; blue, 0 }  ][line width=0.75]    (10.93,-3.29) .. controls (6.95,-1.4) and (3.31,-0.3) .. (0,0) .. controls (3.31,0.3) and (6.95,1.4) .. (10.93,3.29)   ;
\draw   (225.72,51.59) .. controls (231.16,57.15) and (230.77,66.35) .. (224.85,72.15) .. controls (218.93,77.94) and (209.72,78.14) .. (204.28,72.58) .. controls (198.84,67.02) and (199.23,57.82) .. (205.15,52.02) .. controls (211.07,46.22) and (220.28,46.03) .. (225.72,51.59) -- cycle ; \draw   (225.72,51.59) -- (204.28,72.58) ; \draw   (224.85,72.15) -- (205.15,52.02) ;
\draw   (227.72,125.59) .. controls (233.16,131.15) and (232.77,140.35) .. (226.85,146.15) .. controls (220.93,151.94) and (211.72,152.14) .. (206.28,146.58) .. controls (200.84,141.02) and (201.23,131.82) .. (207.15,126.02) .. controls (213.07,120.22) and (222.28,120.03) .. (227.72,125.59) -- cycle ; \draw   (227.72,125.59) -- (206.28,146.58) ; \draw   (226.85,146.15) -- (207.15,126.02) ;
\draw   (225.72,259.59) .. controls (231.16,265.15) and (230.77,274.35) .. (224.85,280.15) .. controls (218.93,285.94) and (209.72,286.14) .. (204.28,280.58) .. controls (198.84,275.02) and (199.23,265.82) .. (205.15,260.02) .. controls (211.07,254.22) and (220.28,254.03) .. (225.72,259.59) -- cycle ; \draw   (225.72,259.59) -- (204.28,280.58) ; \draw   (224.85,280.15) -- (205.15,260.02) ;
\draw    (214.67,233.17) -- (214.06,252) ;
\draw [shift={(214,254)}, rotate = 271.83] [color={rgb, 255:red, 0; green, 0; blue, 0 }  ][line width=0.75]    (10.93,-3.29) .. controls (6.95,-1.4) and (3.31,-0.3) .. (0,0) .. controls (3.31,0.3) and (6.95,1.4) .. (10.93,3.29)   ;
\draw    (215.67,103.17) -- (215.96,120) ;
\draw [shift={(216,122)}, rotate = 268.99] [color={rgb, 255:red, 0; green, 0; blue, 0 }  ][line width=0.75]    (10.93,-3.29) .. controls (6.95,-1.4) and (3.31,-0.3) .. (0,0) .. controls (3.31,0.3) and (6.95,1.4) .. (10.93,3.29)   ;
\draw    (214.67,26.17) -- (214.96,43) ;
\draw [shift={(215,45)}, rotate = 268.99] [color={rgb, 255:red, 0; green, 0; blue, 0 }  ][line width=0.75]    (10.93,-3.29) .. controls (6.95,-1.4) and (3.31,-0.3) .. (0,0) .. controls (3.31,0.3) and (6.95,1.4) .. (10.93,3.29)   ;
\draw    (228,69) -- (351.06,159.98) ;
\draw [shift={(352.67,161.17)}, rotate = 216.48] [color={rgb, 255:red, 0; green, 0; blue, 0 }  ][line width=0.75]    (10.93,-3.29) .. controls (6.95,-1.4) and (3.31,-0.3) .. (0,0) .. controls (3.31,0.3) and (6.95,1.4) .. (10.93,3.29)   ;
\draw    (232.67,142.17) -- (350.71,167.74) ;
\draw [shift={(352.67,168.17)}, rotate = 192.23] [color={rgb, 255:red, 0; green, 0; blue, 0 }  ][line width=0.75]    (10.93,-3.29) .. controls (6.95,-1.4) and (3.31,-0.3) .. (0,0) .. controls (3.31,0.3) and (6.95,1.4) .. (10.93,3.29)   ;
\draw    (229.67,269.17) -- (353.07,176.37) ;
\draw [shift={(354.67,175.17)}, rotate = 143.06] [color={rgb, 255:red, 0; green, 0; blue, 0 }  ][line width=0.75]    (10.93,-3.29) .. controls (6.95,-1.4) and (3.31,-0.3) .. (0,0) .. controls (3.31,0.3) and (6.95,1.4) .. (10.93,3.29)   ;
\draw  [dash pattern={on 0.84pt off 2.51pt}]  (213.67,155.17) -- (214.67,233.17) ;
\draw   (352.67,168.17) .. controls (352.67,154.36) and (363.86,143.17) .. (377.67,143.17) .. controls (391.47,143.17) and (402.67,154.36) .. (402.67,168.17) .. controls (402.67,181.97) and (391.47,193.17) .. (377.67,193.17) .. controls (363.86,193.17) and (352.67,181.97) .. (352.67,168.17) -- cycle ;
\draw    (376.67,103.17) -- (377.62,141.17) ;
\draw [shift={(377.67,143.17)}, rotate = 268.57] [color={rgb, 255:red, 0; green, 0; blue, 0 }  ][line width=0.75]    (10.93,-3.29) .. controls (6.95,-1.4) and (3.31,-0.3) .. (0,0) .. controls (3.31,0.3) and (6.95,1.4) .. (10.93,3.29)   ;
\draw    (402.67,168.17) -- (443.67,168.17) ;
\draw [shift={(445.67,168.17)}, rotate = 180] [color={rgb, 255:red, 0; green, 0; blue, 0 }  ][line width=0.75]    (10.93,-3.29) .. controls (6.95,-1.4) and (3.31,-0.3) .. (0,0) .. controls (3.31,0.3) and (6.95,1.4) .. (10.93,3.29)   ;
\draw   (447,152.33) .. controls (447,146.17) and (452,141.17) .. (458.17,141.17) -- (505.83,141.17) .. controls (512,141.17) and (517,146.17) .. (517,152.33) -- (517,185.83) .. controls (517,192) and (512,197) .. (505.83,197) -- (458.17,197) .. controls (452,197) and (447,192) .. (447,185.83) -- cycle ;
\draw    (517.67,167.17) -- (556.67,167.17) ;
\draw [shift={(558.67,167.17)}, rotate = 180] [color={rgb, 255:red, 0; green, 0; blue, 0 }  ][line width=0.75]    (10.93,-3.29) .. controls (6.95,-1.4) and (3.31,-0.3) .. (0,0) .. controls (3.31,0.3) and (6.95,1.4) .. (10.93,3.29)   ;
\draw   (559,151.33) .. controls (559,145.17) and (564,140.17) .. (570.17,140.17) -- (640.5,140.17) .. controls (646.67,140.17) and (651.67,145.17) .. (651.67,151.33) -- (651.67,184.83) .. controls (651.67,191) and (646.67,196) .. (640.5,196) -- (570.17,196) .. controls (564,196) and (559,191) .. (559,184.83) -- cycle ;
\draw    (651.67,167.17) -- (668.67,168.06) ;
\draw [shift={(670.67,168.17)}, rotate = 183.01] [color={rgb, 255:red, 0; green, 0; blue, 0 }  ][line width=0.75]    (10.93,-3.29) .. controls (6.95,-1.4) and (3.31,-0.3) .. (0,0) .. controls (3.31,0.3) and (6.95,1.4) .. (10.93,3.29)   ;
\draw    (68,62) -- (112,62) ;
\draw    (112,62) -- (149,39) ;

\draw    (71,136) -- (115,136) ;
\draw    (115,136) -- (152,113) ;

\draw    (70,267) -- (114,267) ;
\draw    (114,267) -- (151,244) ;

\draw (44,48.4) node [anchor=north west][inner sep=0.75pt]   [font=\LARGE]{$X_{t}^{1}$};
\draw (47,120.4) node [anchor=north west][inner sep=0.75pt]   [font=\LARGE]{$X_{t}^{2}$};
\draw (46,256.4) node [anchor=north west][inner sep=0.75pt]   [font=\LARGE] {$X_{t}^{\ell}$};
\draw (219,18.4) node [anchor=north west][inner sep=0.75pt]  [font=\LARGE] {$h_t^{1}$};
\draw (220,95.4) node [anchor=north west][inner sep=0.75pt]  [font=\LARGE]  {$h_t^{2}$};
\draw (217,222.4) node [anchor=north west][inner sep=0.75pt] [font=\LARGE]  {$h_t^{\ell}$};
\draw (363,153.4) node [anchor=north west][inner sep=0.75pt]  [font=\huge]  {$\sum $};
\draw (367,80.4) node [anchor=north west][inner sep=0.75pt]  [font=\LARGE] {$w_{t}$};
\draw (414,143.4) node [anchor=north west][inner sep=0.75pt]  [font=\LARGE]  {$S_{t}$};
\draw (559,159) node [anchor=north west][inner sep=0.75pt]  [font=\Large][align=left] {Thresholding};
\draw (450,149.33) node [anchor=north west][inner sep=0.75pt]  [font=\Large] [align=left] {Envelope \\[-10pt] detector};
\draw (527,145.4) node [anchor=north west][inner sep=0.75pt]  [font=\LARGE] {$Y_{t}$};
\draw (675,158.4) node [anchor=north west][inner sep=0.75pt][font=\LARGE]  {$Z_{t}$};
\draw (124,54.4) node [anchor=north west][inner sep=0.75pt]  [font=\LARGE]{$\beta _{1}$};
\draw (127,128.4) node [anchor=north west][inner sep=0.75pt]   [font=\LARGE]{$\beta _{2}$};
\draw (126,259.4) node [anchor=north west][inner sep=0.75pt] [font=\LARGE]  {$\beta _{\ell}$};

\end{tikzpicture}

}
	\setlength{\belowcaptionskip}{-10pt}
	\caption{Non-coherent $(\ell,k)$-Many Access Channel. }
	\label{fig:qtzn}
	\end{figure}

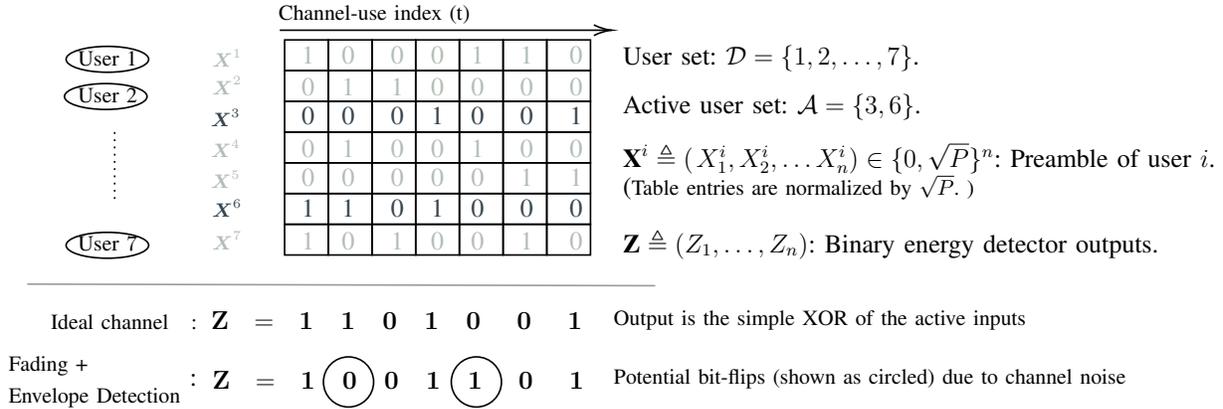
\begin{figure}   
	\centering
	\resizebox{6.5in}{!}{
\tikzset{every picture/.style={line width=0.75pt}} 

\begin{tikzpicture}[x=0.75pt,y=0.75pt,yscale=-1,xscale=1]

\draw  [dash pattern={on 0.84pt off 2.51pt}]  (59,90) -- (59,132) ;
\draw   (28,68) .. controls (28,63.58) and (39.42,60) .. (53.5,60) .. controls (67.58,60) and (79,63.58) .. (79,68) .. controls (79,72.42) and (67.58,76) .. (53.5,76) .. controls (39.42,76) and (28,72.42) .. (28,68) -- cycle ;
\draw   (29,45) .. controls (29,40.58) and (40.42,37) .. (54.5,37) .. controls (68.58,37) and (80,40.58) .. (80,45) .. controls (80,49.42) and (68.58,53) .. (54.5,53) .. controls (40.42,53) and (29,49.42) .. (29,45) -- cycle ;

\draw   (29,160) .. controls (29,155.58) and (40.42,152) .. (54.5,152) .. controls (68.58,152) and (80,155.58) .. (80,160) .. controls (80,164.42) and (68.58,168) .. (54.5,168) .. controls (40.42,168) and (29,164.42) .. (29,160) -- cycle ;

\draw   (164,33) -- (191,33) -- (191,52) -- (164,52) -- cycle ;
\draw   (191,33) -- (218,33) -- (218,52) -- (191,52) -- cycle ;
\draw   (218,33) -- (245,33) -- (245,52) -- (218,52) -- cycle ;
\draw   (245,33) -- (272,33) -- (272,52) -- (245,52) -- cycle ;
\draw   (272,33) -- (299,33) -- (299,52) -- (272,52) -- cycle ;
\draw   (299,33) -- (326,33) -- (326,52) -- (299,52) -- cycle ;
\draw   (326,33) -- (353,33) -- (353,52) -- (326,52) -- cycle ;

\draw   (164,52) -- (191,52) -- (191,71) -- (164,71) -- cycle ;
\draw   (191,52) -- (218,52) -- (218,71) -- (191,71) -- cycle ;
\draw   (218,52) -- (245,52) -- (245,71) -- (218,71) -- cycle ;
\draw   (245,52) -- (272,52) -- (272,71) -- (245,71) -- cycle ;
\draw   (272,52) -- (299,52) -- (299,71) -- (272,71) -- cycle ;
\draw   (299,52) -- (326,52) -- (326,71) -- (299,71) -- cycle ;
\draw   (326,52) -- (353,52) -- (353,71) -- (326,71) -- cycle ;

\draw   (164,71) -- (191,71) -- (191,90) -- (164,90) -- cycle ;
\draw   (191,71) -- (218,71) -- (218,90) -- (191,90) -- cycle ;
\draw   (218,71) -- (245,71) -- (245,90) -- (218,90) -- cycle ;
\draw   (245,71) -- (272,71) -- (272,90) -- (245,90) -- cycle ;
\draw   (272,71) -- (299,71) -- (299,90) -- (272,90) -- cycle ;
\draw   (299,71) -- (326,71) -- (326,90) -- (299,90) -- cycle ;
\draw   (326,71) -- (353,71) -- (353,90) -- (326,90) -- cycle ;

\draw   (164,147) -- (191,147) -- (191,166) -- (164,166) -- cycle ;
\draw   (191,147) -- (218,147) -- (218,166) -- (191,166) -- cycle ;
\draw   (218,147) -- (245,147) -- (245,166) -- (218,166) -- cycle ;
\draw   (245,147) -- (272,147) -- (272,166) -- (245,166) -- cycle ;
\draw   (272,147) -- (299,147) -- (299,166) -- (272,166) -- cycle ;
\draw   (299,147) -- (326,147) -- (326,166) -- (299,166) -- cycle ;
\draw   (326,147) -- (353,147) -- (353,166) -- (326,166) -- cycle ;

\draw   (164,90) -- (191,90) -- (191,109) -- (164,109) -- cycle ;
\draw   (191,90) -- (218,90) -- (218,109) -- (191,109) -- cycle ;
\draw   (218,90) -- (245,90) -- (245,109) -- (218,109) -- cycle ;
\draw   (245,90) -- (272,90) -- (272,109) -- (245,109) -- cycle ;
\draw   (272,90) -- (299,90) -- (299,109) -- (272,109) -- cycle ;
\draw   (299,90) -- (326,90) -- (326,109) -- (299,109) -- cycle ;
\draw   (326,90) -- (353,90) -- (353,109) -- (326,109) -- cycle ;

\draw   (164,109) -- (191,109) -- (191,128) -- (164,128) -- cycle ;
\draw   (191,109) -- (218,109) -- (218,128) -- (191,128) -- cycle ;
\draw   (218,109) -- (245,109) -- (245,128) -- (218,128) -- cycle ;
\draw   (245,109) -- (272,109) -- (272,128) -- (245,128) -- cycle ;
\draw   (272,109) -- (299,109) -- (299,128) -- (272,128) -- cycle ;
\draw   (299,109) -- (326,109) -- (326,128) -- (299,128) -- cycle ;
\draw   (326,109) -- (353,109) -- (353,128) -- (326,128) -- cycle ;

\draw   (164,128) -- (191,128) -- (191,147) -- (164,147) -- cycle ;
\draw   (191,128) -- (218,128) -- (218,147) -- (191,147) -- cycle ;
\draw   (218,128) -- (245,128) -- (245,147) -- (218,147) -- cycle ;
\draw   (245,128) -- (272,128) -- (272,147) -- (245,147) -- cycle ;
\draw   (272,128) -- (299,128) -- (299,147) -- (272,147) -- cycle ;
\draw   (299,128) -- (326,128) -- (326,147) -- (299,147) -- cycle ;
\draw   (326,128) -- (353,128) -- (353,147) -- (326,147) -- cycle ;

\draw    (161,27) -- (364,27) ;
\draw [shift={(366,27)}, rotate = 180] [color={rgb, 255:red, 0; green, 0; blue, 0 }  ][line width=0.75]    (10.93,-3.29) .. controls (6.95,-1.4) and (3.31,-0.3) .. (0,0) .. controls (3.31,0.3) and (6.95,1.4) .. (10.93,3.29)   ;
\draw [color={rgb, 255:red, 155; green, 155; blue, 155 }  ,draw opacity=1 ]   (5,184) -- (393,185) ;
\draw   (188,244.5) .. controls (188,235.94) and (194.94,229) .. (203.5,229) .. controls (212.06,229) and (219,235.94) .. (219,244.5) .. controls (219,253.06) and (212.06,260) .. (203.5,260) .. controls (194.94,260) and (188,253.06) .. (188,244.5) -- cycle ;
\draw   (267,244.5) .. controls (267,235.94) and (273.94,229) .. (282.5,229) .. controls (291.06,229) and (298,235.94) .. (298,244.5) .. controls (298,253.06) and (291.06,260) .. (282.5,260) .. controls (273.94,260) and (267,253.06) .. (267,244.5) -- cycle ;

\draw (173,35) node [anchor=north west][inner sep=0.75pt]  [color={rgb, 255:red, 178; green, 190; blue, 181 }  ,opacity=1 ] [align=left] {1 \ \ 0 \ \ \ 0 \ \ 0 \ \ 1 \ \ \ 1 \ \ \ 0 \ };
\draw (173,55) node [anchor=north west][inner sep=0.75pt]  [color={rgb, 255:red, 178; green, 190; blue, 181 }  ,opacity=1 ] [align=left] {0 \ \ 1 \ \ \ 1 \ \ 0 \ \ 0 \ \ \ 0 \ \ \ 0 \ };
\draw (173,73) node [anchor=north west][inner sep=0.75pt]  [color={rgb, 255:red, 54; green, 69; blue, 79 }  ,opacity=1 ] [align=left] {0 \ \ 0 \ \ \ 0 \ \ 1 \ \ 0 \ \ \ 0 \ \ \ 1 \ };
\draw (173,93) node [anchor=north west][inner sep=0.75pt]  [color={rgb, 255:red, 178; green, 190; blue, 181 }  ,opacity=1 ] [align=left] {0 \ \ 1 \ \ \ 0 \ \ 0 \ \ 1 \ \ \ 0 \ \ \ 0 \ };
\draw (173,111) node [anchor=north west][inner sep=0.75pt]  [color={rgb, 255:red, 178; green, 190; blue, 181 }  ,opacity=1 ] [align=left] {0 \ \ 0 \ \ \ 0 \ \ 0  \ \ 0 \ \ \ 1 \ \ \ 1 \ };
\draw (173,150) node [anchor=north west][inner sep=0.75pt]  [color={rgb, 255:red, 178; green, 190; blue, 181 }  ,opacity=1 ] [align=left] {1 \ \ 0 \ \ \ 1 \ \ 0 \ \ 0 \ \ \ 1 \ \ \ 0 \ };
\draw (173,130) node [anchor=north west][inner sep=0.75pt]  [color={rgb, 255:red, 54; green, 69; blue, 79}  ,opacity=1 ] [align=left] {1 \ \ 1 \ \ \ 0 \ \ 1 \ \ 0 \ \ \ 0 \ \ \ 0 \ };
\draw (35,153) node [anchor=north west][inner sep=0.75pt]  [font=\small] [align=left] {User $\displaystyle 7$};
\draw (35,38) node [anchor=north west][inner sep=0.75pt]  [font=\small] [align=left] {User 1};
\draw (35,61) node [anchor=north west][inner sep=0.75pt]  [font=\small] [align=left] {User 2};
\draw (118,37) node [anchor=north west][inner sep=0.75pt] [color={rgb, 255:red, 178; green, 190; blue, 181 }  ,opacity=1 ] [font=\footnotesize] [align=left] {$\displaystyle \boldsymbol{X}^{1}$};
\draw (118,54) node [anchor=north west][inner sep=0.75pt] [color={rgb, 255:red, 178; green, 190; blue, 181 }  ,opacity=1 ] [font=\footnotesize] [align=left] {$\displaystyle \boldsymbol{X}^{2}$};
\draw (117,74) node [anchor=north west][inner sep=0.75pt] [color={rgb, 255:red, 54; green, 69; blue, 79}  ,opacity=1 ] [font=\footnotesize] [align=left] {$\displaystyle {\boldsymbol{X}^{3}}$};
\draw (118,130) node [anchor=north west][inner sep=0.75pt] [color={rgb, 255:red, 54; green, 69; blue, 79}  ,opacity=1 ] [font=\footnotesize] [align=left] {$\displaystyle {\boldsymbol{X}^{6}}$};
\draw (117,112) node [anchor=north west][inner sep=0.75pt] [color={rgb, 255:red, 178; green, 190; blue, 181 }  ,opacity=1 ] [font=\footnotesize] [align=left] {$\displaystyle \boldsymbol{X}^{5}$};
\draw (117,93) node [anchor=north west][inner sep=0.75pt] [color={rgb, 255:red, 178; green, 190; blue, 181 }  ,opacity=1 ] [font=\footnotesize] [align=left] {$\displaystyle \boldsymbol{X}^{4}$};
\draw (118,150) node [anchor=north west][inner sep=0.75pt][color={rgb, 255:red, 178; green, 190; blue, 181 }  ,opacity=1 ]  [font=\footnotesize] [align=left] {$\displaystyle \boldsymbol{X}^{7}$};
\draw (159,10) node [anchor=north west][inner sep=0.75pt]  [font=\small] [align=left] {Channel-use index (t)};
\draw (366,199) node [anchor=north west][inner sep=0.75pt]   [align=left] {{\small Output is the simple XOR of the active inputs}};
\draw (118,236.4) node [anchor=north west][inner sep=0.75pt]    {$\mathbf{Z\ \ =\ \ 1  \ \ \ 0\ \ \ 0\ \ \ 1\ \ \ 1 \ \ \ \ 0\ \ \ \ 1} \ $};
\draw (-8,227) node [anchor=north west][inner sep=0.75pt]   [align=left] { {\small Fading +}\\[-10pt]{\small Envelope Detection }};
\draw (117,199.4) node [anchor=north west][inner sep=0.75pt]    {$\mathbf{Z\ \ =\ \ 1\ \ \ 1\ \ \ 0\ \ \ 1\ \ \ 0\ \ \ \ 0\ \ \ \ 1} \ $};

\draw (-10,201) node [anchor=north west][inner sep=0.75pt]   [align=left] { \ \ \ \ \  {\small Ideal channel  \ \ :} };
\draw (104,238) node [anchor=north west][inner sep=0.75pt]   [align=left] {:};
\draw (366,235) node [anchor=north west][inner sep=0.75pt]   [align=left] {{\small Potential bit-flips (shown as circled) due to channel noise}};
\draw (366,35) node [anchor=north west][inner sep=0.75pt]   [align=left] {{ User set: $\mathcal{D}=\{1,2,\ldots,7\}$.}};

\draw (366,65) node [anchor=north west][inner sep=0.75pt]   [align=left] {{ Active user set: $\mathcal{A}=\{3,6\}$.}};

\draw (366,95) node [anchor=north west][inner sep=0.75pt]   [align=left] {{ $\textbf{X}^{i} \triangleq (\hspace{0.05cm}X_1^i,X_2^i,\ldots X_{n}^i) \in \{0,\sqrt{P}\}^{n} $: Preamble of user $i$.}};
\draw (366,115) node [anchor=north west][inner sep=0.75pt]   [align=left] {\ {(\small Table entries are normalized by $\sqrt{P}$. )} };
\draw (366,148) node [anchor=north west][inner sep=0.75pt]   [align=left] {{ $\textbf{Z} \triangleq (Z_1,\ldots, Z_n)$: Binary   energy detector outputs. }};
\end{tikzpicture}
}
	\setlength{\belowcaptionskip}{-10pt}
	\caption{Non-coherent $(7,2)$-MnAC and its equivalence to GT. Users 3 and 6 are assumed to be active. }
	\label{fig:qtzn2}
	\end{figure}
Taking into account sparse user activity, we focus on the asymptotic regime where $\ell$ can be unbounded  while the number of active users scale as $k=\Theta(\ell^{\alpha});$  $0 \leq \alpha <1$. This is  consistent with the notion of  Many Access channel (MnAC) model introduced by Chen \textit{et al.} \cite{7852531}  with the difference that instead of our combinatorial setting with $k$ active users, \cite{7852531}  considers the setting where each user can be active  probabilistically.

For purposes of activity detection, each user $i$ is assigned  an i.i.d  binary preamble  $\textbf{X}^{i} \triangleq (\hspace{0.05cm}X_1^i,X_2^i,\ldots X_{n}^i) \in \{0,\sqrt{P}\}^{n} $ of length $n$ generated via a  Bernoulli process with probability  $q_i$, $\forall i \in \mathcal{D}$. The sampling probability vector is defined as $\textbf{q}:=(q_1,q_2,\ldots,q_{\ell})$. We use $\mathcal{P}$ to denote the entire set of preambles. During activity detection,  the active users transmit their OOK-modulated preamble in a time-synchronized manner over $n$ channel-uses with  power $P$ during the `On' symbols. The received symbol $S_t$ during the $t^{th}$ channel-use is  
\begin{equation}
    S_t=(\boldsymbol{\beta} \circ \mathbf{h}_t) \cdot \mathbf{X}_{t}+W_t, \forall t \in\{1,2,\ldots,n\}.
\end{equation}
Here, $\textbf{X}_t=(X^1_t,\ldots, X^\ell_t)$ is the vector representing the $t^{th}$ preamble entries of the $\ell$ users and  $\textbf{h}_t=(h^1_t,\ldots,h^\ell_{t})$ is the vector of channel coefficients of the $\ell$ users during the $t^{th}$ channel-use  where $h_t^i$'s are   i.i.d  $ \mathcal{C} \mathcal{N}\left(0,\sigma^{2}\right), \forall i \in \mathcal{D}, \forall t \in \{1,2,\ldots n\} $.  The AWGN noise vector $\textbf{W}=(W_1,\ldots, W_n)$ is composed of i.i.d entries,  $W_t   \sim \mathcal{C} \mathcal{N}\left(0, \sigma_w^{2}\right)$ and is independent of $\textbf{h}_t$ $\forall t\in \{1,2,\ldots,n\}$.  The vector $\boldsymbol{\beta}=(\beta_{1},\ldots,\beta_{\ell})$ represents the activity status such  that $\beta_{i}=1$, if $i^{th}$ user is active; 0 otherwise.  $\boldsymbol{\beta} \circ \mathbf{h}_t$ denotes the element-wise product between the vectors $\boldsymbol{\beta}$ and $\mathbf{h}_t$. 

We assume that there is no CSI available at the BS since we are operating in a fast fading scenario such as FHSS or  mobile IoT environments \cite{6477839,923716}. i.e., $\textbf{h}_t, \forall t\in \{1,2,\ldots,n\}$ is unknown \textit{a-priori}. Hence, the BS employs non-coherent detection \cite{6120373} as shown in Fig. 1. During the $t^{th}$ channel use, the envelope detector processes  $S_t$ to obtain $Y_{t}=|S_t|$, the envelope of $S_t$.  Thresholding  after envelope detection produces a binary output $Z_{t}$ such that
\begin{equation}
  Z_{t}=1 \text { if } |S_t|^2>\gamma ; \text { else } Z_{t}=0.  
  \label{threshold}
\end{equation} Here, $\gamma $ represents a pre-determined threshold parameter.  

Here onwards, we use the terminology  \textit{non-coherent $(\ell,k)$-MnAC}  to  refer to the constrained version of $(\ell,k)$-MnAC channel model described above where devices are limited by OOK preamble transmissions and base station is limited by non-coherent energy detection.

\subsection{Problem formulation}
Given the preambles $ \mathcal{P} =\{\textbf{X}^i: i\in \mathcal{D}\}$  and the vector $\textbf{Z}=(Z_1,\ldots, Z_n)$  of channel outputs, BS aims to identify the set of active users $\mathcal{A}$, alternatively the vector $\boldsymbol{\beta}$, during the activity detection phase. We use the following definitions:

\begin{definition}
\textbf{ Activity detection}: For a given set of preambles $\mathcal{P} = \{\textbf{X}^i: i\in \mathcal{D}\}$, an activity detection function  $\hat{\boldsymbol{\beta}}:\{0,1\}^{n} \rightarrow \{0,1\}^{\ell}$  is a deterministic rule that maps the binary valued channel outputs $\textbf{Z}$ to an estimate $\hat{\boldsymbol{\beta}}$ of the  activity status vector such that the Hamming weight of $\hat{\boldsymbol{\beta}}$ is $k$.
\end{definition}

\begin{definition}
\textbf{Probability of erroneous identification}: For an activity detection function  $\hat{\boldsymbol{\beta}}$, probability of erroneous identification  $\mathbb{P}_e^{(\ell)}$ is defined as \begin{equation}
     \mathbb{P}_e^{(\ell)}:=\frac{1}{{\ell \choose k}} \sum_{\boldsymbol{\beta}:\sum_{i=1}^{\ell}{\beta_i}=k} \mathbb{P}(\hat{\boldsymbol{\beta}} \neq \boldsymbol{\beta}). \label{errid}
 \end{equation} 
\end{definition}

\begin{definition}
\textbf{Minimum user identification cost} \cite{7852531}: The minimum user identification
cost is said to be $n(\ell)$ if for every $0 < \epsilon <1$, there exists a preamble of length $n_0 = (1 + \epsilon)n(\ell)$ such that the probability
of erroneous identification  tends to 0 as $\ell \rightarrow \infty$, whereas
the error probability is strictly bounded away from 0 if
$n_0 = (1 - \epsilon)n(\ell)$.
\end{definition}

Note that  the minimum user identification cost  $n(\ell)$ is a function of the number of users $\ell$ and is defined based on its asymptotic behaviour as $\ell \rightarrow \infty$. 

 Our goal is to characterize $n(\ell)$  of the non-coherent $(\ell,k)-$MnAC  which would be indicative of the time-efficiency of BS in identifying the active users.

\subsection{GT model for active device identification}

GT problem \cite{8926588} is a sparse inference problem where the objective is to identify a sparse set of \textit{defective} items from a large set of items by performing tests on groups of items. In the error-free scenario, each group test produces a binary output where a \textbf{1} corresponds to the inclusion of at least one defective in the group being tested whereas a \textbf{0} indicates that the tested group of items excludes all defectives. The tests need to be devised such that the defective set of items can be recovered using the binary vector of test outcomes with a minimum number of tests. On a high level, GT models  can be classified as  adaptive and non-adaptive models.  In adaptive GT, the previous test results can be used to design the future tests. In non-adaptive setting, all group tests are designed independent of each other.

To draw the equivalence between active device identification in the non-coherent $(\ell,k)$-MnAC and GT, we can consider the devices in the mMTC network as items and the active devices as defective items. Each channel-use $t \in \{1,2,\ldots, n\}$ corresponds to a group test in which a randomly chosen subset of devices engage in joint transmission of `On' symbols, if active. The random selection of devices during the $n$ testing instances is based on the binary preambles  $\textbf{X}^{i}  \in \{0,\sqrt{P}\}^{n}, \forall i \in \mathcal{D}$. The BS measures received energy during each channel-use and produces a binary 0-1 output $Z_t$ as in  (\ref{threshold}). These 1-bit energy measurements at the receiver side corresponds to the group test results.

In summary, the random binary preambles and the binary energy measurements renders active device identification as a GT problem where testing corresponds to energy detection on the received signal envelope. This is illustrated in Fig. \ref{fig:qtzn} and Fig. \ref{fig:qtzn2}. Note that active device identification in  non-coherent $(\ell,k)$-MnAC is a noisy GT problem since channel noise and fading can probabilistically lead to a 0 or 1 energy detector output even in the presence or absence of active devices respectively. Moreover, our framework corresponds to non-adaptive GT since the testing/grouping pattern is based on predefined binary preambles and is not updated based on the GT results of previous channel-uses.
We will observe in Section \ref{sec3} that the GT model can be leveraged to provide an achievability scheme to the minimum user identification cost for the non-coherent $(\ell,k)-$MnAC when $\alpha = 0 $, ie., $k= O(1)$ and a corresponding lower bound when $\alpha \neq 0 $.

	\section{minimum user identification cost}
	\label{sec3}
	
 In this section, we derive the minimum user identification cost for the non-coherent $(\ell,k)-$ MnAC in the $k=\Theta(1)$ regime as presented in Theorem \ref{thm3}. First, we provide an equivalent characterization of the active device identification problem in  the non-coherent $(\ell,k)-$ MnAC by viewing it as a  decoding problem in a point-to-point  vector input - scalar  output  channel whose inputs correspond to the active users as shown in Fig. \ref{fig:redalpha}. Thereafter, we derive the maximum rate of the equivalent channel by exploiting its cascade structure. Then, we use  GT theory \cite{8926588,6157065} to relate the maximum rate with the minimum user identification cost of non-coherent $(\ell,k)-$MnAC.  Here onwards, we drop the subscript $t$ denoting channel-use index. All logarithms are in base 2.

\subsection{Equivalent  Channel Model} 
 \vspace{0.05cm}
The aim of active device identification  in non-coherent $(\ell,k)-$ MnAC is to identify the set of active devices based on the binary detector output vector $\textbf{Z}$ and the preamble set $\mathcal{P}$.  This can be alternatively viewed as a decoding problem in a traditional point to point communication channel as follows:  The source message is the activity status vector $\boldsymbol{\beta}$ denoting the set of active users $\mathcal{A}$ and the codebook is the entire set of preambles $\mathcal{P}$. Given $\boldsymbol{\beta}$ for an active set $\mathcal{A}=\{a_1,\ldots a_k\}$, the codeword becomes the corresponding subset of preambles $\{\textbf{X}^{i}: i \in \mathcal{A}\}$. Note that since the preambles for each user $i \in \mathcal{A}$ are independently designed using a $Bern(q_i)$ random variable, our equivalent channel model has the additional constraint that the only tunable parameters available for codebook design are the sampling probabilities $\{{q}_i: i \in \mathcal{D}\}$. Hence, identifying active devices reduces to the problem of decoding messages in this equivalent constrained channel. Clearly, the minimum user identification cost for the non-coherent $(\ell,k)-$MnAC is at least the number of channel-uses required to communicate the set of active users; i.e., $\log{\ell\choose k}$ bits of information over this equivalent constrained  point-to-point channel. We formally setup this equivalent channel model below.

Considering the channel in Fig. 1, since the set of active users is $\mathcal{A}=\{a_1,\ldots a_k\}$, the input to the non-coherent $(\ell,k)-$MnAC  is $\Tilde{\textbf{X}}=(X^{a_1},\ldots X^{a_k})$.   Thus,  the signal at the input of envelope detector is 
    $ S=\sqrt{P}\sum_{m=1}^{k}h^{a_m} +W.$
 Let $V=\frac{\sum_{i=1}^{k}X^{a_i}}{\sqrt{P}}$ denote the Hamming weight of $\Tilde{\textbf{X}}$ which is the number of active users transmitting `On' signal during the particular  channel-use we have at hand.
Thus, conditioned on $V=v$, $U:=|S|^{2}$ follows an exponential  distribution given by
\begin{equation}
     f_{U|V}(u|v)=\frac{1}{v \sigma^{2}P+\sigma_{w}^{2}} e^{-\frac{u}{v \sigma^{2}P+\sigma_{w}^{2}}}, u \geq 0.
     \label{expdis}
 \end{equation}  As evident from (\ref{threshold}) and (\ref{expdis}), we have $\Tilde{\textbf{X}} \rightarrow V\rightarrow Z$, i.e.,   the transition probability $p\left(z\mid \Tilde{\textbf{x}},v\right)$  is only dependent on the  channel input  $\Tilde{\textbf{X}}=(X^{a_1},X^{a_2},\ldots X^{a_k})$  through its Hamming weight  $V$.  Also,  since $  f_{U|V}$ is an exponential distribution, $p_v:=p\left(Z=0\mid V=v\right)$ can be expressed as
\begin{equation}
p_v=1-e^{-\frac{\gamma}{v \sigma^{2}P+\sigma_{w}^{2}}}. \label{channeleq}
\end{equation} Similarly, $\operatorname{Pr}\left(Z=1\mid V=v\right)=1-p_v.$ Note that $p_v$ is a strictly decreasing function of $v$ where $v \in\{0,1,..,k\}$ assuming  w.l.o.g. positive values for $\sigma^2, \sigma_w^2$ and $\gamma$.

Thus, the non-coherent $(\ell,k)-$MnAC can be equivalently viewed as a traditional point-to-point communication channel whose input is  the active user set as in Fig. \ref{fig:redalpha}. Moreover, this equivalent channel can be modeled as a cascade of two channels; the first channel computes the Hamming weight $V$ of the input $\Tilde{\textbf{X}}$ whereas the second channel translates the Hamming weight $V$ into a binary output $Z$ depending on the fading statistics $(\sigma^2)$, noise variance $(\sigma_w^2)$ of the wireless channel and the non-coherent detector threshold $\gamma$. We exploit this cascade channel structure to establish the minimum user identification cost $n(\ell)$ for the non-coherent $(\ell,k)-$MnAC.

\begin{figure}   
	\centering
	\resizebox{4.5in}{!}{

\tikzset{every picture/.style={line width=1.4pt}} 

\begin{tikzpicture}[x=.9pt,y=.85pt,yscale=-1,xscale=1.1]

\draw   (59.92,86.98) -- (92.32,86.98) -- (92.32,119.45) -- (59.92,119.45) -- cycle ;
\draw   (59.92,131.79) -- (92.32,131.79) -- (92.32,167.63) -- (59.92,167.63) -- cycle ;
\draw   (59.92,221.39) -- (92.32,221.39) -- (92.32,257.23) -- (59.92,257.23) -- cycle ;
\draw  [dash pattern={on 1.84pt off 2.51pt}]  (76.25,175.25) -- (76.77,216.83) ;
\draw    (92.82,145.23) -- (145.76,145.95) ;
\draw [shift={(147.76,145.97)}, rotate = 180.78] [color={rgb, 255:red, 0; green, 0; blue, 0 }  ][line width=0.75]    (10.93,-3.29) .. controls (6.95,-1.4) and (3.31,-0.3) .. (0,0) .. controls (3.31,0.3) and (6.95,1.4) .. (10.93,3.29)   ;
\draw    (91.83,239.31) -- (144.78,240.03) ;
\draw [shift={(146.78,240.06)}, rotate = 180.78] [color={rgb, 255:red, 0; green, 0; blue, 0 }  ][line width=0.75]    (10.93,-3.29) .. controls (6.95,-1.4) and (3.31,-0.3) .. (0,0) .. controls (3.31,0.3) and (6.95,1.4) .. (10.93,3.29)   ;
\draw    (92.82,107.59) -- (145.76,108.31) ;
\draw [shift={(147.76,108.34)}, rotate = 180.78] [color={rgb, 255:red, 0; green, 0; blue, 0 }  ][line width=0.75]    (10.93,-3.29) .. controls (6.95,-1.4) and (3.31,-0.3) .. (0,0) .. controls (3.31,0.3) and (6.95,1.4) .. (10.93,3.29)   ;
\draw   (155.99,80.56) -- (262.26,80.56) -- (262.26,259.77) -- (155.99,259.77) -- cycle ;
\draw    (263.57,96.84) -- (285.26,96.7) ;
\draw [shift={(287.26,96.69)}, rotate = 179.64] [color={rgb, 255:red, 0; green, 0; blue, 0 }  ][line width=0.75]    (10.93,-3.29) .. controls (6.95,-1.4) and (3.31,-0.3) .. (0,0) .. controls (3.31,0.3) and (6.95,1.4) .. (10.93,3.29)   ;
\draw    (263.57,116.55) -- (285.26,116.42) ;
\draw [shift={(287.26,116.4)}, rotate = 179.64] [color={rgb, 255:red, 0; green, 0; blue, 0 }  ][line width=0.75]    (10.93,-3.29) .. controls (6.95,-1.4) and (3.31,-0.3) .. (0,0) .. controls (3.31,0.3) and (6.95,1.4) .. (10.93,3.29)   ;
\draw    (263.57,250.06) -- (285.26,249.92) ;
\draw [shift={(287.26,249.91)}, rotate = 179.64] [color={rgb, 255:red, 0; green, 0; blue, 0 }  ][line width=0.75]    (10.93,-3.29) .. controls (6.95,-1.4) and (3.31,-0.3) .. (0,0) .. controls (3.31,0.3) and (6.95,1.4) .. (10.93,3.29)   ;
\draw  [dash pattern={on 0.84pt off 2.51pt}]  (310.95,137.91) -- (310.95,147.77) -- (310.95,158.52) ;
\draw  [dash pattern={on 0.84pt off 2.51pt}]  (280.95,50.91) -- (280.95,248.52) ;
\draw   (468.87,104.91) .. controls (468.87,92.53) and (479.92,82.5) .. (493.55,82.5) .. controls (507.18,82.5) and (518.22,92.53) .. (518.22,104.91) .. controls (518.22,117.28) and (507.18,127.31) .. (493.55,127.31) .. controls (479.92,127.31) and (468.87,117.28) .. (468.87,104.91) -- cycle ;
\draw   (470.85,235.73) .. controls (470.85,223.35) and (481.89,213.32) .. (495.52,213.32) .. controls (509.15,213.32) and (520.2,223.35) .. (520.2,235.73) .. controls (520.2,248.1) and (509.15,258.13) .. (495.52,258.13) .. controls (481.89,258.13) and (470.85,248.1) .. (470.85,235.73) -- cycle ;
\draw    (321.81,91.32) -- (473.41,249.36) ;
\draw [shift={(474.79,250.81)}, rotate = 226.19] [color={rgb, 255:red, 0; green, 0; blue, 0 }  ][line width=0.75]    (10.93,-3.29) .. controls (6.95,-1.4) and (3.31,-0.3) .. (0,0) .. controls (3.31,0.3) and (6.95,1.4) .. (10.93,3.29)   ;
\draw    (322.79,118.2) -- (473.29,249.49) ;
\draw [shift={(474.79,250.81)}, rotate = 221.1] [color={rgb, 255:red, 0; green, 0; blue, 0 }  ][line width=0.75]    (10.93,-3.29) .. controls (6.95,-1.4) and (3.31,-0.3) .. (0,0) .. controls (3.31,0.3) and (6.95,1.4) .. (10.93,3.29)   ;
\draw    (324.77,250.81) -- (472.79,250.81) ;
\draw [shift={(474.79,250.81)}, rotate = 180] [color={rgb, 255:red, 0; green, 0; blue, 0 }  ][line width=0.75]    (10.93,-3.29) .. controls (6.95,-1.4) and (3.31,-0.3) .. (0,0) .. controls (3.31,0.3) and (6.95,1.4) .. (10.93,3.29)   ;
\draw    (322.79,118.2) -- (472.82,91.66) ;
\draw [shift={(474.79,91.32)}, rotate = 169.97] [color={rgb, 255:red, 0; green, 0; blue, 0 }  ][line width=0.75]    (10.93,-3.29) .. controls (6.95,-1.4) and (3.31,-0.3) .. (0,0) .. controls (3.31,0.3) and (6.95,1.4) .. (10.93,3.29)   ;
\draw    (321.81,91.32) -- (472.79,91.32) ;
\draw [shift={(474.79,91.32)}, rotate = 180] [color={rgb, 255:red, 0; green, 0; blue, 0 }  ][line width=0.75]    (10.93,-3.29) .. controls (6.95,-1.4) and (3.31,-0.3) .. (0,0) .. controls (3.31,0.3) and (6.95,1.4) .. (10.93,3.29)   ;
\draw    (324.77,250.81) -- (473.42,92.77) ;
\draw [shift={(474.79,91.32)}, rotate = 133.25] [color={rgb, 255:red, 0; green, 0; blue, 0 }  ][line width=0.75]    (10.93,-3.29) .. controls (6.95,-1.4) and (3.31,-0.3) .. (0,0) .. controls (3.31,0.3) and (6.95,1.4) .. (10.93,3.29)   ;
\draw   (301.08,80.56) -- (319.83,80.56) -- (319.83,259.77) -- (301.08,259.77) -- cycle ;
\draw  [dash pattern={on 0.84pt off 2.51pt}]  (310.95,200.63) -- (310.95,210.49) -- (309.96,233.78) ;
\draw    (322.79,171.96) -- (473.03,92.25) ;
\draw [shift={(474.79,91.32)}, rotate = 152.05] [color={rgb, 255:red, 0; green, 0; blue, 0 }  ][line width=0.75]    (10.93,-3.29) .. controls (6.95,-1.4) and (3.31,-0.3) .. (0,0) .. controls (3.31,0.3) and (6.95,1.4) .. (10.93,3.29)   ;
\draw    (322.79,171.96) -- (473.02,249.89) ;
\draw [shift={(474.79,250.81)}, rotate = 207.42] [color={rgb, 255:red, 0; green, 0; blue, 0 }  ][line width=0.75]    (10.93,-3.29) .. controls (6.95,-1.4) and (3.31,-0.3) .. (0,0) .. controls (3.31,0.3) and (6.95,1.4) .. (10.93,3.29)   ;
\draw  [fill={rgb, 255:red, 155; green, 155; blue, 155 }  ,fill opacity=0.23 ] (119.47,103.11) .. controls (119.47,79.36) and (138.72,60.1) .. (162.48,60.1) -- (402.17,60.1) .. controls (425.93,60.1) and (445.18,79.36) .. (445.18,103.11) -- (445.18,232.14) .. controls (445.18,255.89) and (425.93,275.15) .. (402.17,275.15) -- (162.48,275.15) .. controls (138.72,275.15) and (119.47,255.89) .. (119.47,232.14) -- cycle ;

\draw (165,104.4) node [anchor=north west][inner sep=0.75pt]  [font=\huge]  {$\frac{\sum _{i=1}^{k}X_{t}^{a_{i}}}{\sqrt{P}}$};
\draw (156,189) node [anchor=north west][inner sep=0.75pt]  [font=\large] [align=left] {{\large \ \ \ Hamming  }\\[-10pt] {\large \ \ \ weight $(V)$}\\[-10pt] {\large \  computation}};
\draw (28,32) node [anchor=north west][inner sep=0.75pt]  [font=\small] [align=left] {{\large Channel input}\\{\large  $\displaystyle \ \ \ \ \ \ \ \ \Tilde{\textbf{X}}$}};
\draw (220.53,30.96) node [anchor=north west][inner sep=0.75pt]  [font=\large]  {{Hamming weight}{ $V$}};
\draw (487.53,95.96) node [anchor=north west][inner sep=0.75pt]  [font=\LARGE]  {$0$};
\draw (488.53,225.39) node [anchor=north west][inner sep=0.75pt]  [font=\LARGE]  {$1$};
\draw (421,32) node [anchor=north west][inner sep=0.75pt]  [font=\large] [align=left] {{\large Channel Output}\\{\large $\displaystyle \ \ \ \ \ \ \ \ \ \ Z$}};
\draw (48.2,284.73) node [anchor=north west][inner sep=0.75pt]  [font=\large]  {$active\ user\ set:\ \{a_{1} ,a_{2} ,\dotsc ,a_{k} \}$};
\draw (306.92,168.09) node [anchor=north west][inner sep=0.75pt]    {$i$};
\draw (328.51,150.22) node [anchor=north west][inner sep=0.75pt]  [font=\large,rotate=-332.77]  {$p_{i}$};
\draw (329.93,176.77) node [anchor=north west][inner sep=0.75pt]  [font=\large,rotate=-29.2]  {$1-p_{i}$};
\draw (49.02,302.55) node [anchor=north west][inner sep=0.75pt]  [font=\large]  {$X^{a_{i}} \in \{0,\sqrt{P}\}$};
\draw (348.88,282.21) node [anchor=north west][inner sep=0.75pt]  [font=\Large]  {$p_{i} =1-e^{-\frac{\gamma}{i\sigma ^{2}P +\sigma _{w}^{2}}}$};
\draw (61.92,90.38) node [anchor=north west][inner sep=0.75pt]  [font=\large]  {$X^{a_{1}}$};
\draw (60.03,226.07) node [anchor=north west][inner sep=0.75pt]  [font=\large]  {$X^{a_{k}}$};
\draw (60.74,135.53) node [anchor=north west][inner sep=0.75pt]  [font=\large]  {$X^{a_{2}}$};
\draw (306.92,240.67) node [anchor=north west][inner sep=0.75pt]    {$k$};
\draw (305.94,110.75) node [anchor=north west][inner sep=0.75pt]    {$1$};
\draw (305.94,88.35) node [anchor=north west][inner sep=0.75pt]    {$0$};
\end{tikzpicture}
}	\setlength{\belowcaptionskip}{-38pt}
	\caption{Equivalent channel with only the active users of the $(\ell,k)-$MnAC as inputs.}
	\label{fig:redalpha}
	\end{figure}
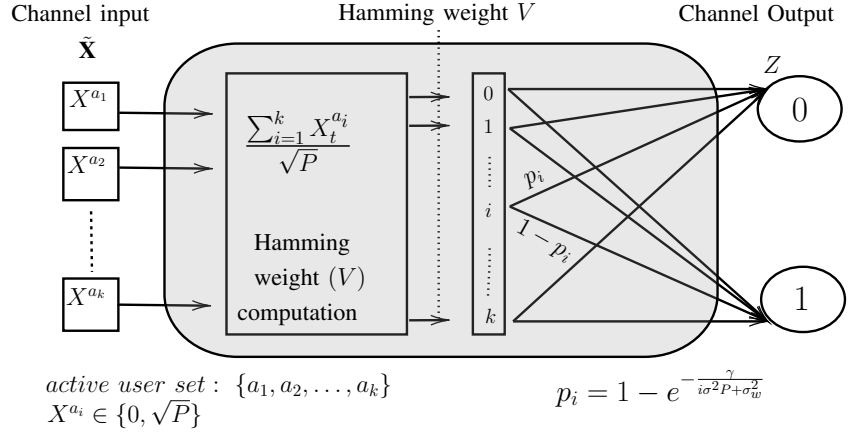

\subsection{ Maximum rate of the equivalent channel}
   
\begin{lemma}
For a  given fading statistics $\sigma^2$, noise variance $\sigma_w^2$, and non-coherent detector threshold $\gamma$ for the $(\ell,k)$-MnAC, the maximum rate of the equivalent point-to-point channel in Fig. \ref{fig:redalpha} is 
 \begin{equation} 
    C=\max_{(\gamma,q_{sp})}h\Big(E\Big[e^{-\frac{\gamma}{V  \sigma^{2}P+\sigma_{w}^{2}}}\Big]\Big)-E\Big[h\Big(e^{-\frac{\gamma}{V  \sigma^{2}P+\sigma_{w}^{2}}}\Big)\Big]  \label{jengap2}
\end{equation} where $E(\cdot)$ denotes expectation w.r.t  $ V$,  $h(x)=-x \log x-(1-x) \log (1-x)$ is the binary entropy function and $q_{sp}$ denotes the optimal sampling probability used for i.i.d preamble generation across all users.
 \label{lem1}
 \end{lemma}
\begin{proof} The equivalent point-to-point communication channel in Fig. \ref{fig:redalpha}  has two tunable parameters, viz, the sampling probability vector $\textbf{q} =\{q_{a_1} \ldots q_{a_k}\}$ at the user side and the non-coherent detector threshold $\gamma$ at the BS. Thus, the maximum rate of this equivalent channel between the binary vector input $\Tilde{\textbf{X}}$ and binary scalar output $Z$ is  
  \begin{equation}
    C=\max_{(\gamma,\textbf{q})} I(\Tilde{\textbf{X}};Z) \label{sumcap}
    \end{equation} where $\textbf{q} =\{q_{a_1} \ldots q_{a_k}\}$ are sampling probabilities of independent Bernoulli processes that users  $\{a_1 \ldots a_k\}$ employ for preamble generation, i.e., $X^{a_i}  \sim Bern(q_{a_i})$  such that  $X^{a_i}\ind X^{a_j}  $ $\forall i,j \in  \{1,\ldots k\}; i\neq j$. 
    
    Symmetry of the problem  with respect to the entries of $\Tilde{\textbf{X}}$ and concavity of $I(\Tilde{\textbf{X}};Z)$ as a function of $p(\Tilde{\textbf{x}})$  suggest that  for optimality we need $q_{a_1}=\ldots q_{a_k} =q_{sp}. $  Thus, $V$ follows Binomial distribution denoted by $V \sim  Bin(k,q_{sp}).$  Using (\ref{sumcap}), $\Tilde{\textbf{X}} \rightarrow V\rightarrow Z$ and that $V$ is a deterministic function of $\Tilde{\textbf{X}}$, we have
    \begin{equation}
     C =\max_{(\gamma,q_{sp} )} I(V;Z) \text{ where } V \sim Bin(k,q_{sp}). 
\label{eq20}
\end{equation} 
   We can write $H(Z)=h(q_s)$ where  $q_s$ is  \begin{equation}
    q_s= E\left( e^{-\frac{\gamma}{i  \sigma^{2}P+\sigma_{w}^{2}}}\right)=\sum_{i=0}^{k} {k\choose i} (1-q_{sp})^{k-i} (q_{sp})^{i} \times e^{-\frac{\gamma}{i  \sigma^{2}P+\sigma_{w}^{2}}}.
\end{equation}
 Also, we have 
\begin{equation*} \small
H(Z \mid V)=\sum_{i=0}^{k}{k\choose i}(1-q_{sp})^{k-i} (q_{sp})^{i} \times h\Big(e^{-\frac{\gamma}{i  \sigma^{2}P+\sigma_{w}^{2}}}\Big)\hspace{.5cm}
  \end{equation*}
  \begin{equation}
    =E\left[h\Big(e^{-\frac{\gamma}{V  \sigma^{2}P+\sigma_{w}^{2}}}\Big)\right]. \hspace{1.9cm} 
  \end{equation}
 \normalsize
Thus, jointly optimizing over $\gamma$ and $q_{sp}$, the maximum rate is given by 
\begin{equation*}
    C=\max_{(\gamma,q_{sp})}h\Big(E\Big[e^{-\frac{\gamma}{V  \sigma^{2}P+\sigma_{w}^{2}}}\Big]\Big)-E\Big[h\Big(e^{-\frac{\gamma}{V  \sigma^{2}P+\sigma_{w}^{2}}}\Big)\Big] ,
\end{equation*}  completing the proof.\end{proof}

Note that the minimum user identification  cost $n(\ell)$ for the non-coherent $(\ell,k)-$MnAC,  must obey $ n(\ell) \geq \frac{\log {\ell\choose k}}{C}$. However,  we cannot directly  claim that  $n(\ell)$  exactly matches the number of channel-uses required to send $\log{\ell\choose k}$ bits of information over the equivalent channel at maximum rate $C$. The reason is that the effective  codeword corresponding to a user set $\mathcal{A}$ is the subset of preambles $\{\textbf{X}^{i}: i \in \mathcal{A}\}$ indexed by the set $\mathcal{A}$.  Thus, the effective codewords of different user sets are potentially dependent due to  overlapping preambles. This is in contrast to the independent codeword assumptions typical to channel capacity analysis thereby rendering the traditional random coding approach ineffective.  We close this gap by making use of the equivalence of active device identification and GT to propose a decoding scheme achieving the minimum user identification cost as presented below.

\vspace{-0.3cm}\subsection{Minimum user identification cost}
\begin{theorem}
The minimum  user identification cost $n(\ell)$ of the non-coherent  $(\ell,k)-MnAC$ with $k=\Theta(1)$  is given by
\begin{equation*} \small
    n(\ell)= \frac{k\log(\ell)}{\max_{(\gamma,q_{sp})}h\Big(E\Big[e^{-\frac{\gamma}{V  \sigma^{2}P+\sigma_{w}^{2}}}\Big]\Big)-E\Big[h\Big(e^{-\frac{\gamma}{V  \sigma^{2}P+\sigma_{w}^{2}}}\Big)\Big]}
\end{equation*} \normalsize where $E(.)$ denotes the expectation w.r.t  $V \sim \operatorname{Bin}(k, q_{sp})$, $q_{sp}$ is the sampling probability for i.i.d preamble generation and $\gamma$ denotes the threshold for non-coherent energy detection.
\label{thm3}
\end{theorem}

\begin{proof} Using standard capacity arguments \cite{10.5555/1146355}  along with Lemma $\ref{lem1}$, it is evident that user identification cost in the non-coherent $(\ell,k)$-MnAC must obey $ n(\ell) \geq \frac{\log {\ell\choose k}}{C}$ where $C$ is as given in Lemma \ref{lem1}. Thus, in the $k =\Theta(1)$ regime we have
\begin{equation} 
    n(\ell) \geq \frac{k\log(\ell)}{\max_{(\gamma,q_{sp})}h\Big(E\Big[e^{-\frac{\gamma}{V  \sigma^{2}P+\sigma_{w}^{2}}}\Big]\Big)-E\Big[h\Big(e^{-\frac{\gamma}{V  \sigma^{2}P+\sigma_{w}^{2}}}\Big)\Big]}.\label{equ}
\end{equation}
 
 Next, we present a decoder as shown in Algorithm \ref{alg:MYALG}  which can  achieve the equality  in $(\ref{equ})$ relying on classical thresholding techniques in information theory \cite{1057459}.  Due to the equivalence of active user identification and GT, we use an  adaptation of  the GT decoding strategy  in  \cite{6157065} to the context of user identification in  non-coherent $(\ell,k)$-MnAC.

\begin{algorithm}
\caption{Decoder achieving the minimum user identification cost in the $k=\Theta(1)$ regime. }\label{alg:cap}
\begin{algorithmic}[1]
\State \textbf{Given}: Set of preambles, $\mathcal{P} \triangleq \{\textbf{X}^i: i\in \mathcal{D}\}$;
Energy detector outputs: $\textbf{Z}\triangleq(Z_1,\ldots, Z_n)$; \hspace{.2cm} $k$: Number of active devices.
\State \textbf{Initialize}: Set of all size-$k$ subsets of $\mathcal{D}$ denoted by $\mathcal{S} =\{\mathcal{K}_1,\mathcal{K}_2,\ldots,\mathcal{K}_{{\ell \choose k}}\}$; Estimate of activity status vector denoted by $\hat{\boldsymbol{\beta}}=\mathbf{0}_{1 \times \ell};$  $\mathcal{A}_{est}=\Phi.$
\For{$i =1: \left|\mathcal{S}\right|$ }
\State $\hat{\mathcal{A}}= \mathcal{K}_i$.
\State Number of threshold tests passed, $n_t=0$.
\State \textbf{Threshold testing}:
\For{$j =1:k$}
\State \textbf{Partition: }$\Lambda^{\{j\}}=\left\{\left(\mathcal{I}^0, \mathcal{I}^1\right): \mathcal{I}^0 \bigcap \mathcal{I}^1=\emptyset, \mathcal{I}^0 \bigcup \mathcal{I}^1=\hat{\mathcal{A}},\left|\mathcal{I}^0\right|=j,\left|\mathcal{I}^1\right|=k-j\right\}$
\State \textbf{Set threshold values: }$\eta_{j}=\log _{2} \frac{\rho}{k{k \choose j}{\ell-k \choose j}}, \forall j \in \{1,2,\ldots,k\}$ for some $\rho >0$.
\If   {$\log_2 \left( \frac{ p\left(\textbf{z} \mid \textbf{x}({I_0}),\textbf{x}({I_1})\right)}{ p\left(\textbf{z} \mid \textbf{x}({I_1})\right)}\right) >\eta_{j}, \forall (I_0,I_1) \in \Lambda^{\{j\}}$}
\State $n_t=n_t+{k \choose j}$.
\EndIf
\EndFor
\If{$n_t =2^k-1$}
$\mathcal{A}_{est}=\mathcal{A}_{est} \bigcup \hat{\mathcal{A}}. $
\EndIf
\EndFor
\If{$\left|\mathcal{A}_{est}\right|=k$}
\State The estimated active device set is $\mathcal{A}_{est}$. 
 \State $\hat{\boldsymbol{\beta}}:\hat{{\beta_i}}=1, \forall i \in \mathcal{A}_{est}$.
\ElsIf {$\left|\mathcal{A}_{est}\right|=0 \text{ OR} \left|\mathcal{A}_{est}\right|>1$}
\State Declare \textbf{ error}.
\EndIf
\end{algorithmic}
\label{alg:MYALG}
\end{algorithm}

The algorithm operates as follows: Given the set of preambles $\mathcal{P}$ and the detector outputs $\textbf{Z}$, the decoder exhaustively searches for an approximate maximum likelihood (ML) size-$k$ set from  all possible ${\ell \choose k}$ subsets. The decoder declares $\mathcal{A}^*$ as the active  set iff
\begin{equation}
    p\left(\textbf{z} \mid \textbf{x}({\mathcal{A}^*})\right)>p\left(\textbf{z} \mid \textbf{x}({\hat{\mathcal{A}} })\right) ; \quad \forall \hat{\mathcal{A}}  \neq \mathcal{A}^*,
\end{equation}  where $\textbf{x}(\mathbf{\mathcal{A}})$ denotes $\{\textbf{x}^{a_i}:a_i\in \mathcal{A}\}$.  The corresponding activity status vector is  ${\boldsymbol{\beta}^*}:{{\beta_i}}=1, \forall i \in \mathcal{A}^*$. We perform this search  through a series of threshold tests for each candidate set $\hat{\mathcal{A}} \in \mathcal{S}$. For each $\hat{\mathcal{A}}$, we consider a partition $( \mathcal{I}^0, \mathcal{I}^1)$  of the candidate active set into its overlap and difference with the true active set as in step (6). i.e., $ \mathcal{I}^1:=\mathcal{A} \cap  \hat{\mathcal{A}} $ and $ \mathcal{I}^0:=\hat{\mathcal{A}} \setminus {\mathcal{A}} $. Then, we perform $2^k-1$ threshold tests as below. 
\begin{equation}
    \log_2 \left( \frac{ p\left(\textbf{z} \mid \textbf{x}({I_0}),\textbf{x}({I_1})\right)}{ p\left(\textbf{z} \mid \textbf{x}({I_1})\right)}\right) >\eta_{s}, \forall (I_0,I_1) \text{ partitions of } \hat{\mathcal{A}},\nonumber \label{deco}
\end{equation}
for predetermined thresholds $\eta_{s}$ where $s:=|I_0| \in \{1,2,\ldots,k\}$. Here, we  used $\textbf{x}(\mathbf{G})$ to denote $\{\textbf{x}^{a_i}:a_i\in \mathbf{G}\}$.

We declare $\hat{\mathcal{A}}$ as the  active user set if it is the only set passing all threshold tests. An error occurs if the active user set fails at least one of the tests in  (\ref{deco}) or an incorrect set passes all the tests.  
As shown in \cite{6157065}, using a set of thresholds given by $ \eta_{s}=\log \frac{\rho}{k{k \choose s}{\ell-k \choose s}}$ for some $\rho >0$,
 we can upper bound $\mathbb{P}_e^{(\ell)}$ as \begin{equation*} 
    \mathbb{P}_e^{(\ell)} \leq \sum_{s=1}^{k}{k \choose s} \mathbb{P}\Bigg( \log \left( \frac{ p\left(\textbf{z} \mid \textbf{x}( \mathcal{I}^0),\textbf{x}( \mathcal{I}^1)\right)}{ p\left(\textbf{z} \mid \textbf{x}( \mathcal{I}^1)\right)}\right)\leq \log  \frac{\rho}{k{k \choose s}{\ell-k \choose s}}\Bigg)+\rho. \label{nonasym}
\end{equation*}

Since  preamble design uses i.i.d Bernoulli random variables, using concentration inequalities \cite{8926588} for the $k=\Theta(1)$ regime, we can achieve $\mathbb{P}_e^{(\ell)}\rightarrow 0,$ with $\ell \rightarrow \infty$ as long as $n(\ell)$ obeys
\begin{equation}
     n(\ell) \geq \max _{s=1, \ldots, k} \frac{\log{\ell-k \choose s}}{I(\textbf{X}({ \mathcal{I}^0}_s);Z \mid \textbf{X}({ \mathcal{I}^1}_s) )}(1+o(1)).\label{teq}
\end{equation} Here, $\textbf{X}({ \mathcal{I}^0}_s):=\{\textbf{X}({ \mathcal{I}^0}):| \mathcal{I}^0|=s\}$ and $\textbf{X}({ \mathcal{I}^1}_s):=\{\textbf{X}({ \mathcal{I}^1}):| \mathcal{I}^0|=s\}$.
Using the asymptotic expression
$
\log {\ell -k \choose s}=s \log \ell+O(1)
$ in the $k=\Theta(1)$ regime implies 
\begin{equation}
     n(\ell) \geq \max _{s=1, \ldots, k} \frac{s\log \ell}{I(\textbf{X}({ \mathcal{I}^0}_s);Z \mid \textbf{X}({ \mathcal{I}^1}_s) )}(1+o(1)).\label{teq1}
\end{equation}
We can rewrite $I_s:=I(\textbf{X}({ \mathcal{I}^0}_s);Z \mid \textbf{X}({ \mathcal{I}^1}_s) )$   as 
\begin{equation} 
    I_s=\sum_{j=k-s+1}^{k}\Big(H\big(\textbf{X}(j)\big)-H\big(\textbf{X}(j) \mid Z, \mathbf{X}(1,2,..,j-1)\big)\Big) \label{con}.
\end{equation}Note that $\frac{I_s}{s}$ attains its minimum when $s=k$ since conditioning reduces entropy \cite{10.5555/1146355}. Thus, when $k=\Theta(1)$, (\ref{teq1}) simplifies to 

$ n(\ell) \geq \frac{k\log(\ell)}{I(\textbf{X};Z  )}(1+o(1)).$  Using the maximum rate of the equivalent channel $C$  in Lemma 1, we can conclude that,  our proposed decoder has $\mathbb{P}_e^{(\ell)} \rightarrow 0$ with $\ell \rightarrow \infty$ as long as 
 \begin{equation*}
    n(\ell) \geq \frac{k\log(\ell)}{\max_{(\gamma,q_{sp})}h\Big(E\Big[e^{-\frac{\gamma}{V  \sigma^{2}P+\sigma_{w}^{2}}}\Big]\Big)-E\Big[h\Big(e^{-\frac{\gamma}{V  \sigma^{2}P+\sigma_{w}^{2}}}\Big)\Big]}(1+o(1)),
\end{equation*} 
completing the proof.

\end{proof}
  
\noindent \textbf{Remark 1.} Note that  $\frac{k(1-\alpha)\log_2(\ell)}{C}$ where $C$ is as given in Lemma 1  is a valid lower bound for $k=\Theta(\ell^\alpha); 0\leq \alpha<1$. This is because user identification requires communicating  at least $\log_2{\ell \choose k}=k \log _{2} \frac{\ell}{k}+O(k)$ bits which corresponds to at least $\frac{k(1-\alpha)\log_2(\ell)}{C}$ number of channel-uses. As shown in Theorem 1, this lower bound is tight when $k = \Theta(1)$ or $\alpha = 0.$

\noindent \textbf{Remark 2.} The minimum user identification cost $n_0(\ell)$ for a Gaussian MnAC with average power constraint $P_{av}$ as described by  Chen \textit{et al}. in \cite{7852531} obeys   \begin{equation}
    n_{0}(\ell)= \frac{k\log(\ell)}{\frac{1}{2}\log (1+ kP_{av})} < n(\ell), \nonumber
\end{equation} where $n(\ell)$ is as given in Theorem 1. This is because $C$ in Lemma 1 is smaller than  $\frac{1}{2}\log (1+ kP_{av})$ where $P_{av}:= q_{sp}^*P$, with $ q_{sp}^*$ being  the optimal $ q_{sp}$ in Theorem 1. Clearly, the user identification in a Gaussian MnAC requires lesser number of channel uses due to the fact that  the model does not incorporate fading and is not constrained to OOK signaling and non-coherent detection as in our non-coherent $(\ell,k)-$MnAC. 
\section{Practical schemes for user identification}

The decoder achieving the minimum user identification cost presented in Algorithm \ref{alg:MYALG} is practically infeasible as its implementation requires exhaustive search over all candidate active user sets of size $k$. In this section, we present several  practical schemes for active user identification in a  non-coherent $(\ell,k)-$MnAC  by viewing it as a decoding problem in an equivalent point to point communication channel as illustrated in Section III.A. Our proposed schemes are primarily  based on N-COMP and BP algorithms, which although have been extensively studied  for compressed sensing problems \cite{6120391,4797638,661103}, were not used in the context of non-coherent $(\ell,k)-$MnAC setting. We are also interested in evaluating how close our proposed practical algorithms can come to
achieving the theoretical minimum user identification cost presented in Theorem 1.

In addition to exact recovery (Def. 2), we  consider the performance of our proposed strategies under a partial recovery criterion where instead of identifying all the $k$ active devices correctly, we are willing to tolerate some amount of misclassifications where an inactive  is declared as active and vice-versa. We formally define the partial recovery setting as follows. 

\begin{definition}
\textbf{ $\zeta \% -$ partial recovery}: For a true active set $\mathcal{A}$ and an estimated active set $\hat{\mathcal{A}}$,  consider an error event $E_1$  defined as 
\begin{equation}
    E_1 := \left\{\left|\hat{\mathcal{A}}^{\mathrm{c}} \cap \mathcal{A}\right| >  k\left(1-\frac{\zeta}{100}\right) \right\}. \label{errev}
\end{equation} We have $\zeta \%$-partial recovery, if \begin{equation}
    \mathbb{P}_{e,\zeta}^{(\ell)}:= P( E_1)  \rightarrow 0 \text { as } \ell \rightarrow \infty, \nonumber
\end{equation} 
where $\mathbb{P}_{e,\zeta}^{(\ell)}$ denotes the probability of error for $\zeta \% -$ partial recovery.  Probability  of successful identification for $\zeta \% -$  recovery is defined as $\mathbb{P}_{succ,\zeta}^{(\ell)}:= 1 -\mathbb{P}_{e,\zeta}^{(\ell)}.$

\end{definition}

The error event $E_1$ considers if the fraction of true active devices in $\mathcal{A}$  that are misdetected exceeds $\left(1-\frac{\zeta}{100}\right)$. Note that since our proposed strategies output a set of size $k$, both false positives and misdetections occur in equal numbers. Thus, the error event $E_1$ in (\ref{errev}) is essentially equivalent to $\left\{\left|\hat{\mathcal{A}}^{\mathrm{c}} \cap \mathcal{A}\right| > k\left(1-\frac{\zeta}{100}\right)  \bigcup \left|\hat{\mathcal{A}} \cap \mathcal{A}^c\right| >  k\left(1-\frac{\zeta}{100}\right)  \right\}$. Moreover, when $\zeta = 100$, the partial recovery setting boils down to our original setting of exact recovery given in Def. 2. The case when the number of active users $k$ is unknown at the BS will be discussed in Section IV.C.

   \subsection{N-COMP based user identification}

N-COMP is a decoding strategy for noisy GT proposed by Chan \textit{et al.} in \cite{6120391}. The basic idea is to declare an item  to be defective or non-defective based on the fraction of positive tests they participate in \cite{5394787}. This significantly simplifies the run-time and storage requirements in comparison to the optimal exhaustive search strategy in Algorithm 1.  As established in Section III, the active device identification framework in non-coherent $(\ell,k)-$MnAC can be viewed as a GT problem. This motivates us to study the performance of N-COMP approach in identifying active devices in the non-coherent $(\ell,k)-$MnAC. Note that in the existing literature, the N-COMP GT decoder was primarily analyzed  in the context of symmetric noise
models wherein errors in test results are equally likely \cite{6120373,6763117}. In contrast, the active device identification problem as modeled by the equivalent channel in Fig.3 has an asymmetric noise  and hence requires further investigations to validate its performance. Another work which considers N-COMP GT decoder is \cite{4797638} where the focus is on the problem of neighbor discovery in a wireless sensor network rather than the massive random access setting considered in this paper.

Let $\mathcal{G}_i:=\frac{\sum_{t=1}^{n}X^{i}_t}{\sqrt{P}}$ denote the Hamming weight of $\textbf{X}^{i}$, the binary preamble assigned to user $i$. In other words, $\mathcal{G}_i$ indicates the number of channel uses in which the user $i$ transmits `On' symbols if it is in active state. Let $\mathcal{R}_i:=\frac{\sum_{t=1}^{n}X^{i}_tZ_t}{\sqrt{P}}$ denote the number of these channel uses in which the received energy at the detector exceeds a predetermined threshold. In N-COMP based user identification, our strategy is to classify the $k$ users with highest value for $\frac{\mathcal{R}_i}{\mathcal{G}_i}$ as active.

We conducted  simulations of the N-COMP strategy to analyze the number of channel-uses required for identifying the active devices successfully. We considered two different success criterion, viz, exact recovery in which all active devices are to be correctly identified and 90$\%$ recovery ($\zeta =90)$ in which we require atleast 90$\%$  of the active devices to be correctly identified. In other words, with  90$\%$ recovery, we are willing to tolerate upto  10$\%$ of false positives and false negatives. Fig. 4 and Fig. 5 show how the probability of successful identification varies with the number of channel-uses for  $\ell=1000$ users with $k=25$. We considered different values for the SNR (defined as $\frac{P\sigma^2}{\sigma_w^2}$) ranging from 0 dB to 10 dB. As expected, with increase in SNR, N-COMP decoding performance gets better leading to a faster identification of active devices in the network. From Fig. 5, one can infer that relaxing the success criteria to 90$\%$ recovery improves the probability of successful identification significantly.

\begin{figure} 
\centering
\begin{minipage}{.5\textwidth}
	\centering
		\begin{tikzpicture}
  \sbox0{\includegraphics[width=.9\linewidth,height=65mm,trim={1.3cm 0.55cm 0 0},clip]{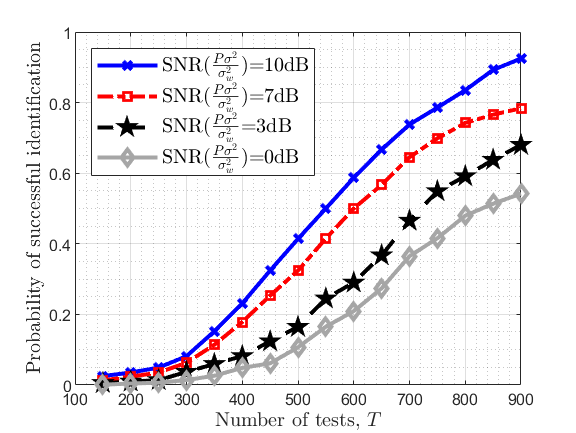}}
  \node[above right,inner sep=0pt] at (0,0)  {\usebox{0}};
  \node[black] at (0.5\wd0,-0.06\ht0) {\normalsize{Number of channel-uses, $n$}};
  \node[black,rotate=90] at (-0.04\wd0,0.5\ht0) {\normalsize{Probability of successful identification}};
\end{tikzpicture}
	\setlength{\belowcaptionskip}{-15pt}
	\caption{NCOMP: Exact recovery for $(1000,25)$-MnAC. }
	\label{fig:2usercap1}
	\end{minipage}%
	\begin{minipage}{.5\textwidth}
	\centering
		\begin{tikzpicture}
  \sbox0{\includegraphics[width=.9\linewidth,height=65mm,trim={1.3cm 0.55cm 0 0},clip]{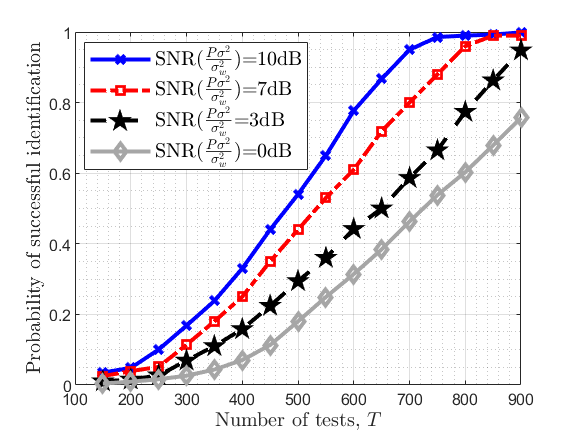}}
  \node[above right,inner sep=0pt] at (0,0)  {\usebox{0}};
  \node[black] at (0.5\wd0,-0.06\ht0) {\normalsize{Number of channel-uses, $n$}};
  \node[black,rotate=90] at (-0.04\wd0,0.5\ht0) {\normalsize{Probability of successful identification}};
\end{tikzpicture}
	\setlength{\belowcaptionskip}{-15pt}
	\caption{NCOMP: 90$\%$ recovery for $(1000,25)$-MnAC. }
	\label{fig:2usercap2}
	\end{minipage}
	\end{figure}
	\subsection{BP based user identification}
	
	BP is an efficient method  for solving sparse inference problems by iteratively propagating messages over the edges of the factor graph \cite{5169989}. BP is known to produce good approximations to  the marginals of the posterior distribution of the variables in the factor graph. In wireless communications, BP is well known in the context of iterative decoding for Turbo and LDPC codes \cite{661103}. BP is also used in GT decoding, which is essentially a  Boolean sparse inference problem \cite{5707018, hara2022sparse}. In \cite{5707018}, the authors considered  BP based GT decoding under a noise model that combines the addition and
dilution models. In the context of mMTC, BP has been previously studied as a means for unsourced random access the users share
the same codebook and
the decoder is only required to provide an unordered list of
user messages \cite{9654225}.
	
	Recognizing the suitability of BP in sparse recovery problems including GT, we aim to employ a BP based strategy for the problem of active device identification in the non-coherent $(\ell,k)-$MnAC. In BP based user identification, we  estimate the set of active devices by  approximately performing  bitwise-MAP detection through message passing on a bipartite graph with devices on one side and the binary energy measurements on the other side \cite{5707018}.    In contrast to the previous studies in \cite{5707018}, our noise model as given in $(\ref{channeleq})$ incorporating channel fading  effects and non-coherent detection at the BS side is more intricate. Note that   in \cite{8849288},  Polyanskiy  proposed an alternating BP scheme as a candidate practical solution for random access in a  quasi-static fading multiple access channel where each user’s transmissions are attenuated by random fading gains, unknown to the receiver. However, \cite{8849288} considers  iterative coding scheme based on LDPC codes and slow fading scenarios whereas our work focus on energy-efficient OOK  signaling in a fast fading non-coherent $(\ell,k)-$MnAC.

	To identify the activity status vector $\boldsymbol{\beta}$ using the binary vector of energy detector outputs $\textbf{Z}=(Z_1,\ldots, Z_n)$, BP relies on an approximate Maximum Aposteriori Probability (MAP) based decoding where instead of  maximizing the posterior as a function of all the $\beta_i$'s, we use the marginals of the posterior distribution. i.e., $\widehat{\beta_i}=\underset{\beta_i \in\{0,1\}}{\arg \max }\hspace{2pt} \mathbb{P}\left(\beta_i \mid \mathbf{Z}\right)$. This approximation drastically reduces the search space corresponding to the original MAP optimization. $\widehat{\beta_i}$ can be rewritten as
\begin{equation}
 \hat{\beta}_i= \underset{\beta_i \in\{0,1\}}{\arg\max } \sum_{\sim \beta_i}\left[\prod_{t=1}^n p\left(Z_t \mid \boldsymbol{\beta}\right) \prod_{j=1}^\ell p\left(\beta_j\right)\right].
 \label{bpexp} 
\end{equation} Eqn. (\ref{bpexp}) can be approximately computed using loopy BP based on a bipartite graph with $\ell$ devices at one side and the $n$ energy detector  outputs at the other side. The priors $p\left(\beta_j\right)$'s are initialized to $\frac{k}{\ell}$ and updated through message passing in further iterations.
Note that \begin{equation}
    p(Z_t \mid \boldsymbol{\beta})=p\left(Z_t \mid V \right),
\end{equation} since $Z_t$ depends on $\boldsymbol{\beta}$ only through $V$ (defined in Section III.A), the number of active devices transmitting `On' symbols during the $t^{th}$ channel-use. This symmetry as illustrated in our equivalent channel model in Fig.3 is key in  guaranteeing an efficient computation of BP message update rules \cite{5707018}. We employ  BP with Soft Thresholding (BP-ST) in which we classify the $k$ users with the highest marginals as active. Similar to the NCOMP strategy, we observed that BP-ST decoding performance shows improvement with SNR (plots are omitted for brevity).

\begin{figure}   
	\centering
\begin{minipage}{.5\textwidth}
	\begin{tikzpicture}
  \sbox0{\includegraphics[width=.9\linewidth,height=65mm,trim={1.4cm 0.8cm 0 0},clip]{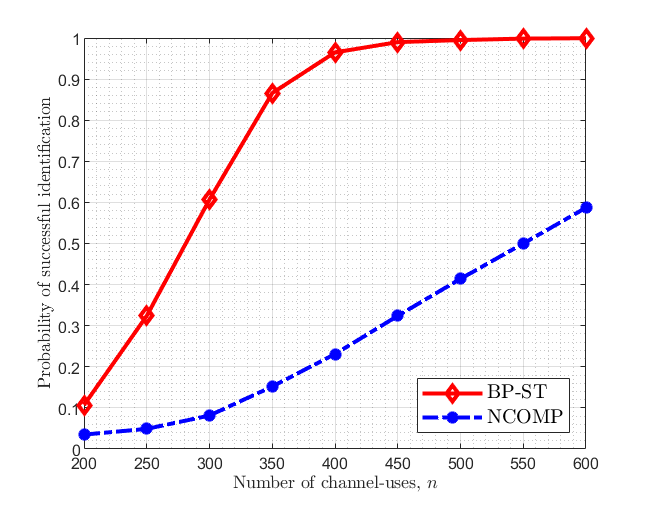}}
  \node[above right,inner sep=0pt] at (0,0)  {\usebox{0}};
  \node[black] at (0.5\wd0,-0.06\ht0) {\normalsize{Number of channel-uses, $n$}};
  \node[black,rotate=90] at (-0.04\wd0,0.5\ht0) {\small{Probability of successful identification}};
\end{tikzpicture}
	\setlength{\belowcaptionskip}{-15pt} 
	\caption{\footnotesize{Exact recovery: $(1000,25)$-MnAC at SNR = 10 dB.}}
	\label{fig:z5}
	\end{minipage}%
\begin{minipage}{.5\textwidth}
	\centering
	\begin{tikzpicture}
  \sbox0{\includegraphics[width=.9\linewidth,height=65mm,trim={1.4cm 0.8cm 0 0},clip]{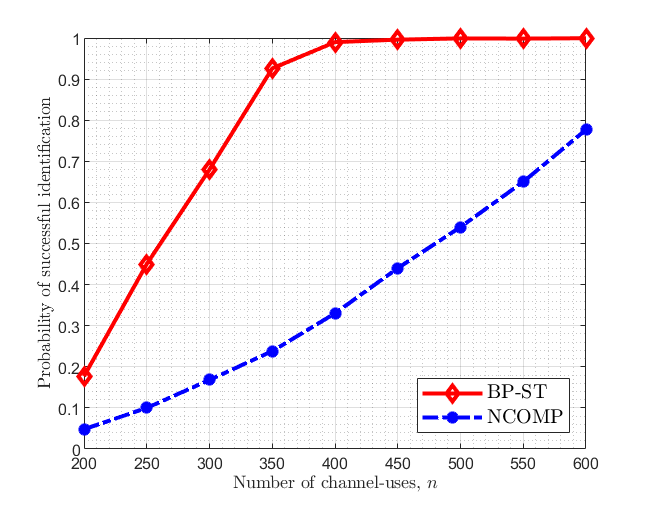}}
  \node[above right,inner sep=0pt] at (0,0)  {\usebox{0}};
  \node[black] at (0.5\wd0,-0.06\ht0) {\normalsize{Number of channel-uses, $n$}};
  \node[black,rotate=90] at (-0.04\wd0,0.5\ht0) {\small{Probability of successful identification}};
\end{tikzpicture}
\setlength{\belowcaptionskip}{-15pt} 
	\caption{\footnotesize{$90\%$ recovery: $(1000,25)$-MnAC at SNR = 10 dB.} }
	\label{fig:z6}
	\end{minipage}
	\end{figure}
In Fig. \ref{fig:z5} and Fig. \ref{fig:z6}, we compare the  probability of successful identification in a  $(1000,25)$-MnAC at SNR = 10 dB for BP-ST and NCOMP strategies under the two success  criterion, viz, exact and $90\%$ recovery. As evident, BP shows significant  performance gains over NCOMP which translates to faster identification of active devices. Moreover, the curves indicate that if we are willing to relax the success criterion by setting $\zeta$ to a lower value like $90\%$, both strategies tend to perform better while maintaining the relative order in terms of probability of successful identification. Although N-COMP significantly  lags behind BP, owing to its simplicity and non-iterative nature, it still can be used in applications where run-time and hardware complexity is a  critical factor.

\begin{figure}   
	\centering
	\begin{tikzpicture}
  \sbox0{\includegraphics[width=.55\linewidth,height=70mm,trim={1.1cm 0.8cm 0 0},clip]{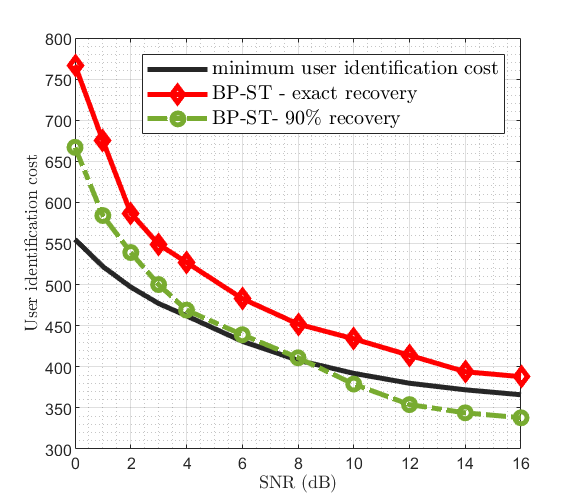}}
  \node[above right,inner sep=0pt] at (0,0)  {\usebox{0}};
  \node[black] at (0.5\wd0,-0.06\ht0) {\small{SNR}};
  \node[black,rotate=90] at (-0.04\wd0,0.5\ht0) {\small{Number of channel-uses $n$}};
\end{tikzpicture}
	\setlength{\belowcaptionskip}{-25pt} 
	\caption{ Minimum user identification cost vs BP-ST for $(1000,25)$-MnAC. }
	\label{fig:z4}
	\end{figure}

In Fig. \ref{fig:z4}, we compare the minimum user identification cost in Theorem 1 with the performance of BP-ST. Evidently, the gap between the BP-ST curve and the minimum user identification cost reduces notably in the high SNR regime. We also observe that  BP can approach and even exceed the minimum user identification cost if we relax the success criterion slightly to $90\%$ recovery by allowing room for a small number of  misdetections.

\subsection{BP-based user identification when $k$ is unknown}	
Since simulation results indicate that BP is more efficient compared to NCOMP, we restrict our focus to BP in this section. To take into account the uncertainty associated with the size $\hat{k}$ of the estimated active set $\hat{\mathcal{A}}$, we consider  an alternate definition for  the error event different from  (\ref{errev}) as below. 

\begin{definition}
\textbf{ $\zeta \% -$ Partial recovery with unknown k}: For a true active set $\mathcal{A}$ and an estimated active set $\hat{\mathcal{A}}$,  consider an error event $E_2$  defined as 
\begin{equation}
    E_2:= \left\{k(1-\sigma)\leq \hat{k} \leq k(1+\sigma)\right\}^{\mathrm{c}} \bigcup E_1, 
    \label{errev1}
    \end{equation} where $E_1$ is as given in $(\ref{errev})$. We have $\zeta \%$-partial recovery with unknown $k$, if \begin{equation}
    \mathbb{P}_{e,\zeta}^{(\ell)}:= P( E_2)  \rightarrow 0 \text { as } \ell \rightarrow \infty, \nonumber
\end{equation} 
where $\mathbb{P}_{e,\zeta}^{(\ell)}$ denotes the probability of error for $\zeta \% -$ partial recovery.  Probability  of successful identification for $\zeta \% -$  recovery is defined as $\mathbb{P}_{succ,\zeta}^{(\ell)}:= 1 -\mathbb{P}_{e,\zeta}^{(\ell)}.$

\end{definition}
Note that $\mathbb{P}(E_2) \geq \mathbb{P}(E_1)$ due to  the first term $\left\{k(1-\sigma)\leq \hat{k} \leq k(1+\sigma)\right\}^{\mathrm{c}}$ in (\ref{errev1}). It implies that in addition to $E_1$, if the estimated active set has a size deviating from the size of the true active set by more than a fraction of $\sigma$, we declare an error. In the case when $k$ is known \textit{apriori}, we can set $\sigma = 0$ and (\ref{errev1}) simplifies to (\ref{errev}). 

    We consider the performance of several loopy BP based strategies for non-coherent active user identification when $k$ is unknown as follows: 
\begin{itemize}
    \item BP with Symmetric Hard Thresholding (BP-SHT) : Here, we classify user $i$ as active if the estimate of $\mathbb{P}\left(\beta_i=\right.$ $1 \mid \mathbf{Z})$ exceeds $\frac{1}{2}$.
    \item
    BP with Asymmetric Hard Thresholding (BP-AHT) : Here,
    we threshold  the  posterior probability $\mathbb{P}\left(\beta_i=\right.$ $1 \mid \mathbf{Z})$ at a  value $\eta  \neq \frac{1}{2}$ allowing more flexibility in controlling false positives and misdetections.

\end{itemize}

In our simulations, we set $\sigma = 0.1$ allowing upto $10\%$ deviation in the value of $\hat{k}$ w.r.t $k$. Firstly, similar to NCOMP and  BP-ST,  we observed that the performance of BP-SHT and BP-AHT also show an increasing trend with SNR (plots omitted for brevity). Hence, here onwards, we focus on a particular SNR value (10 dB) for convenience. 

In Fig. \ref{fig:q4}, we compare the exact recovery ($\zeta = 100$) performance of different BP strategies in terms of how the probability of successful identification varies with the number of channel-uses at an SNR of $10$ dB. Clearly, BP-ST demonstrates superior performance in comparison to BP-SHT and BP-AHT due  to the fact that BP-ST has the knowledge of $k$ at its disposal which enables better  decoding. The corresponding comparison curves for the  $90\%$ recovery case  ($\zeta = 90$) is shown in Fig. \ref{fig:q5}. Similar to the exact recovery case, BP-ST performs better than the other two strategies. However, the gap between the different strategies is less pronounced for  $90\%$ recovery  implying that the lack of knowledge of $k$ is less crucial in partial-recovery scenarios. 

\begin{figure}   
	\centering
	\begin{tikzpicture}
  \sbox0{\includegraphics[width=.55\linewidth,height=70mm,trim={1.4cm 0.8cm 0 0},clip]{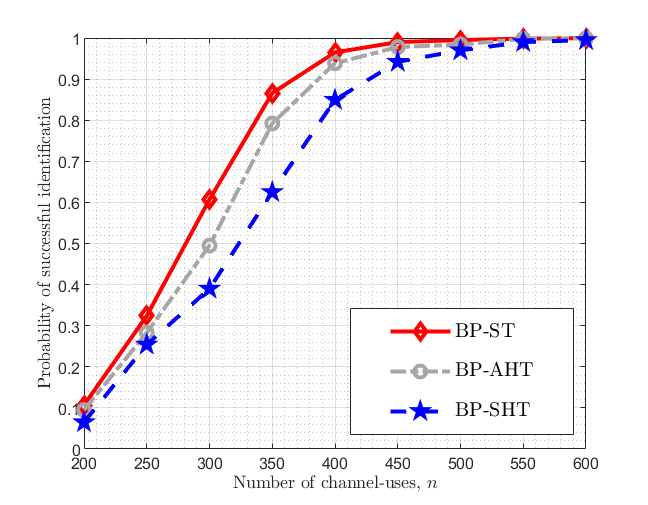}}
  \node[above right,inner sep=0pt] at (0,0)  {\usebox{0}};
  \node[black] at (0.5\wd0,-0.06\ht0) {\small{\small{Number of channel-uses $n$}}};
  \node[black,rotate=90] at (-0.04\wd0,0.5\ht0) {\small{Probability of successful identification}};
\end{tikzpicture}
	\setlength{\belowcaptionskip}{-25pt} 
	\caption{ Comparison of various BP algorithms for exact recovery in $(1000,25)$-MnAC at SNR $=$ 10 dB. }
	\label{fig:q4}
	\end{figure}
\begin{figure}   
	\centering
	\begin{tikzpicture}
  \sbox0{\includegraphics[width=.55\linewidth,height=70mm,trim={1.4cm 0.8cm 0 0},clip]{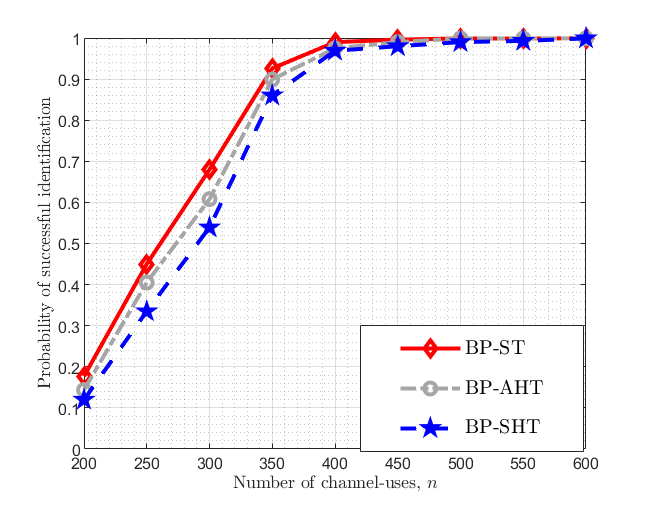}}
  \node[above right,inner sep=0pt] at (0,0)  {\usebox{0}};
  \node[black] at (0.5\wd0,-0.06\ht0) {\small{\small{Number of channel-uses $n$}}};
  \node[black,rotate=90] at (-0.04\wd0,0.5\ht0) {\small{Probability of successful identification}};
\end{tikzpicture}
	\setlength{\belowcaptionskip}{-25pt} 
	\caption{ Comparison of various BP algorithms for $90\%$ recovery in $(1000,25)$-MnAC at  SNR $=$ 10 dB. }
	\label{fig:q5}
	\end{figure}

\section{Conclusion}
	In this paper, we have proposed an active device identification framework for non-coherent $(\ell,k)-$MnAC  in which, during each channel-use the active devices engage in joint on-off preamble transmission  and base station measures the received energy using a threshold-based  binary energy detector. We have demonstrated that the active device identification problem in non-coherent $(\ell,k)-$MnAC can be modeled as a decoding problem in a constrained point-to-point communication channel. Exploiting this equivalence, we have evaluated the performance of our proposed scheme  in terms of the minimum user identification cost - the  minimum number of channel-uses required for reliable active user identification. Specifically, we have established a lower bound for the  minimum user identification cost in the $k=\Theta(\ell^{\alpha})$ which is tight in the $k=\Theta(1)$ regime. We have also illustrated an achievability scheme for $k =\Theta(1)$  based on an approximate ML decoding approach. Considering the computational complexity of this achievability scheme, we have proposed several practical strategies for device identification based on NCOMP and  BP strategies. We have also considered various scenarios in our simulations including partial recovery of active devices and unknown number of active devices. Through several numerical simulations, we have established that BP based strategies can outperform NCOMP strategies significantly in terms of the user identification cost. In fact, it was noted that if we are willing to tolerate a few number of misdetections, BP performance can even surpass the minimum user identification cost.

\bibliographystyle{IEEEtranTCOM}

\bibliography{IEEE_TCOM}

\begin{thebibliography}{10}
\baselineskip 12pt
\providecommand{\url}[1]{#1}
\csname url@samestyle\endcsname
\providecommand{\newblock}{\relax}
\providecommand{\bibinfo}[2]{#2}
\providecommand{\BIBentrySTDinterwordspacing}{\spaceskip=0pt\relax}
\providecommand{\BIBentryALTinterwordstretchfactor}{4}
\providecommand{\BIBentryALTinterwordspacing}{\spaceskip=\fontdimen2\font plus
\BIBentryALTinterwordstretchfactor\fontdimen3\font minus
  \fontdimen4\font\relax}
\providecommand{\BIBforeignlanguage}[2]{{%
\expandafter\ifx\csname l@#1\endcsname\relax
\typeout{** WARNING: IEEEtran.bst: No hyphenation pattern has been}%
\typeout{** loaded for the language `#1'. Using the pattern for}%
\typeout{** the default language instead.}%
\else
\language=\csname l@#1\endcsname
\fi
#2}}
\providecommand{\BIBdecl}{\relax}
\BIBdecl

\bibitem{7736615}
Z.~Dawy, W.~Saad, A.~Ghosh, J.~G. Andrews, and E.~Yaacoub, ``Toward massive
  machine type cellular communications,'' \emph{IEEE Wireless Communications},
  vol.~24, no.~1, pp. 120--128, 2017.

\bibitem{7263368}
A.~Ali, W.~Hamouda, and M.~Uysal, ``Next generation {M2M} cellular networks:
  challenges and practical considerations,'' \emph{IEEE Communications
  Magazine}, vol.~53, no.~9, pp. 18--24, 2015.

\bibitem{9023459}
G.~Gui, M.~Liu, F.~Tang, N.~Kato, and F.~Adachi, ``{6G}: Opening new horizons
  for integration of comfort, security, and intelligence,'' \emph{IEEE Wireless
  Communications}, vol.~27, no.~5, pp. 126--132, 2020.

\bibitem{5937958}
A.~Mezghani and J.~A. Nossek, ``Power efficiency in communication systems from
  a circuit perspective,'' in \emph{2011 IEEE International Symposium of
  Circuits and Systems (ISCAS)}, 2011, pp. 1896--1899.

\bibitem{8534548}
D.~Sundman, M.~M. Lopez, and L.~R. Wilhelmsson, ``Partial on-off keying - a
  simple means to further improve iot performance,'' in \emph{2018 Global
  Internet of Things Summit (GIoTS)}, 2018, pp. 1--5.

\bibitem{9596585}
X.~Guo, L.~Shangguan, Y.~He, J.~Zhang, H.~Jiang, A.~A. Siddiqi, and Y.~Liu,
  ``Efficient ambient lora backscatter with on-off keying modulation,''
  \emph{IEEE/ACM Transactions on Networking}, vol.~30, no.~2, pp. 641--654,
  2022.

\bibitem{6120373}
P.~Zhang, F.~M.~J. Willems, and L.~Huang, ``Investigations of noncoherent {OOK}
  based schemes with soft and hard decisions for {WSN}s,'' in \emph{2011 49th
  Annual Allerton Conference on Communication, Control, and Computing
  (Allerton)}, 2011, pp. 1702--1709.

\bibitem{8386824}
W.~Zhou, Y.~Jia, A.~Peng, Y.~Zhang, and P.~Liu, ``The effect of iot new
  features on security and privacy: New threats, existing solutions, and
  challenges yet to be solved,'' \emph{IEEE Internet of Things Journal},
  vol.~6, no.~2, pp. 1606--1616, 2019.

\bibitem{mendez2018internet}
D.~Mendez~Mena, I.~Papapanagiotou, and B.~Yang, ``Internet of things: Survey on
  security,'' \emph{Information Security Journal: A Global Perspective},
  vol.~27, no.~3, pp. 162--182, 2018.

\bibitem{7036828}
G.-Y. Chang, J.-F. Huang, and Z.-H. Wu, ``A frequency hopping algorithm against
  jamming attacks under asynchronous environments,'' in \emph{2014 IEEE Global
  Communications Conference}, 2014, pp. 324--329.

\bibitem{fi11010016}
\BIBentryALTinterwordspacing
L.~Oliveira, J.~J. P.~C. Rodrigues, S.~A. Kozlov, R.~A.~L. Rabêlo, and V.~H.
  C.~d. Albuquerque, ``Mac layer protocols for internet of things: A survey,''
  \emph{Future Internet}, vol.~11, no.~1, 2019. [Online]. Available:
  \url{https://www.mdpi.com/1999-5903/11/1/16}
\BIBentrySTDinterwordspacing

\bibitem{4394775}
A.~Mpitziopoulos, D.~Gavalas, G.~Pantziou, and C.~Konstantopoulos, ``Defending
  wireless sensor networks from jamming attacks,'' in \emph{2007 IEEE 18th
  International Symposium on Personal, Indoor and Mobile Radio Communications},
  2007, pp. 1--5.

\bibitem{9136922}
M.~Letafati, A.~Kuhestani, H.~Behroozi, and D.~W.~K. Ng, ``Jamming-resilient
  frequency hopping-aided secure communication for internet-of-things in the
  presence of an untrusted relay,'' \emph{IEEE Transactions on Wireless
  Communications}, vol.~19, no.~10, pp. 6771--6785, 2020.

\bibitem{8403656}
O.~Y. Bursalioglu, Z.~Li, C.~Wang, and H.~Papadopoulos, ``Efficient {C-RAN}
  random access for {IoT} devices: Learning links via recommendation systems,''
  in \emph{2018 IEEE International Conference on Communications Workshops (ICC
  Workshops)}, 2018, pp. 1--6.

\bibitem{TS36.213}
``Physical layer procedures-{TS} 36.213,'' 3rd Generation Partnership Project
  (3GPP), Tech. Rep., Sept. 2017.

\bibitem{9266124}
X.~Shao, X.~Chen, D.~W.~K. Ng, C.~Zhong, and Z.~Zhang, ``Cooperative activity
  detection: Sourced and unsourced massive random access paradigms,''
  \emph{IEEE Transactions on Signal Processing}, vol.~68, pp. 6578--6593, 2020.

\bibitem{ALTURJMAN2017299}
F.~Al-Turjman, ``Price-based data delivery framework for dynamic and pervasive
  {IoT},'' \emph{Pervasive and Mobile Computing}, vol.~42, pp. 299--316, 2017.

\bibitem{9060999}
Y.~{Wu}, X.~{Gao}, S.~{Zhou}, W.~{Yang}, Y.~{Polyanskiy}, and G.~{Caire},
  ``Massive access for future wireless communication systems,'' \emph{IEEE
  Wireless Communications}, vol.~27, no.~4, pp. 148--156, 2020.

\bibitem{9537931}
J.~Choi, J.~Ding, N.-P. Le, and Z.~Ding, ``{G}rant-{F}ree random access in
  machine-type communication: Approaches and challenges,'' \emph{IEEE Wireless
  Communications}, pp. 1--8, 2021.

\bibitem{8734871}
Z.~Chen, F.~Sohrabi, and W.~Yu, ``Multi-cell sparse activity detection for
  massive random access: Massive {MIMO} versus cooperative {MIMO},'' \emph{IEEE
  Transactions on Wireless Communications}, vol.~18, no.~8, 2019.

\bibitem{7952810}
Z.~Chen and W.~Yu, ``Massive device activity detection by approximate message
  passing,'' in \emph{2017 IEEE International Conference on Acoustics, Speech
  and Signal Processing (ICASSP)}, 2017, pp. 3514--3518.

\bibitem{8454392}
L.~Liu, E.~G. Larsson, W.~Yu, P.~Popovski, C.~Stefanovic, and E.~de~Carvalho,
  ``Sparse signal processing for grant-free massive connectivity: A future
  paradigm for random access protocols in the internet of things,'' \emph{IEEE
  Signal Processing Magazine}, vol.~35, no.~5, pp. 88--99, 2018.

\bibitem{5695122}
M.~Bayati and A.~Montanari, ``The dynamics of message passing on dense graphs,
  with applications to compressed sensing,'' \emph{IEEE Transactions on
  Information Theory}, vol.~57, no.~2, pp. 764--785, 2011.

\bibitem{7282735}
X.~Li, S.~Pawar, and K.~Ramchandran, ``Sub-linear time compressed sensing using
  sparse-graph codes,'' in \emph{2015 IEEE International Symposium on
  Information Theory (ISIT)}, 2015, pp. 1645--1649.

\bibitem{8444464}
K.~Senel and E.~G. Larsson, ``Grant-free massive {MTC}-enabled massive mimo: A
  compressive sensing approach,'' \emph{IEEE Transactions on Communications},
  vol.~66, no.~12, pp. 6164--6175, 2018.

\bibitem{8262800}
H.~A. {Inan}, P.~{Kairouz}, and A.~{Ozgur}, ``Sparse group testing codes for
  low-energy massive random access,'' in \emph{2017 55th Annual Allerton
  Conference on Communication, Control, and Computing (Allerton)}, 2017, pp.
  658--665.

\bibitem{9500808}
J.~Robin and E.~Erkip, ``Sparse activity discovery in energy constrained
  multi-cluster {IoT} networks using group testing,'' in \emph{ICC 2021 - IEEE
  International Conference on Communications}, 2021, pp. 1--6.

\bibitem{9593173}
------, ``Access delay constrained activity detection in massive random
  access,'' in \emph{2021 IEEE 22nd International Workshop on Signal Processing
  Advances in Wireless Communications (SPAWC)}, 2021, pp. 191--195.

\bibitem{7852531}
X.~Chen, T.-Y. Chen, and D.~Guo, ``Capacity of {G}aussian many-access
  channels,'' \emph{IEEE Transactions on Information Theory}, vol.~63, no.~6,
  pp. 3516--3539, 2017.

\bibitem{6691257}
X.~Chen and D.~Guo, ``{G}aussian many-access channels: definition and symmetric
  capacity,'' in \emph{2013 IEEE Information Theory Workshop (ITW)}, 2013, pp.
  1--5.

\bibitem{8849288}
S.~S. Kowshik, K.~Andreev, A.~Frolov, and Y.~Polyanskiy, ``Energy efficient
  random access for the quasi-static fading {MAC},'' in \emph{2019 IEEE
  International Symposium on Information Theory (ISIT)}, 2019, pp. 2768--2772.

\bibitem{923716}
I.~Abou-Faycal, M.~Trott, and S.~Shamai, ``The capacity of discrete-time
  memoryless {R}ayleigh-fading channels,'' \emph{IEEE Transactions on
  Information Theory}, vol.~47, no.~4, pp. 1290--1301, 2001.

\bibitem{5174497}
K.~Witrisal, G.~Leus, G.~J. Janssen, M.~Pausini, F.~Troesch, T.~Zasowski, and
  J.~Romme, ``Noncoherent ultra-wideband systems,'' \emph{IEEE Signal
  Processing Magazine}, vol.~26, no.~4, pp. 48--66, 2009.

\bibitem{7514754}
L.~Jing, E.~De~Carvalho, P.~Popovski, and O.~Martínez, ``Design and
  performance analysis of noncoherent detection systems with massive receiver
  arrays,'' \emph{IEEE Transactions on Signal Processing}, vol.~64, no.~19, pp.
  5000--5010, 2016.

\bibitem{8926588}
\BIBentryALTinterwordspacing
M.~{Aldridge}, O.~{Johnson}, and J.~{Scarlett}, \emph{Group Testing: An
  Information Theory Perspective}, 2019. [Online]. Available:
  \url{https://ieeexplore.ieee.org/document/8926588}
\BIBentrySTDinterwordspacing

\bibitem{1096146}
T.~Berger, N.~Mehravari, D.~Towsley, and J.~Wolf, ``Random multiple-access
  communication and group testing,'' \emph{IEEE Transactions on
  Communications}, vol.~32, no.~7, pp. 769--779, 1984.

\bibitem{8945}
D.~Kurtz and M.~Sidi, ``Multiple access algorithms via group testing for
  heterogeneous population of users,'' \emph{IEEE Transactions on
  Communications}, vol.~36, no.~12, pp. 1316--1323, 1988.

\bibitem{1057026}
J.~Wolf, ``Born again group testing: Multiaccess communications,'' \emph{IEEE
  Transactions on Information Theory}, vol.~31, no.~2, pp. 185--191, 1985.

\bibitem{8849823}
H.~A. Inan, S.~Ahn, P.~Kairouz, and A.~Ozgur, ``A group testing approach to
  random access for short-packet communication,'' in \emph{2019 IEEE
  International Symposium on Information Theory (ISIT)}, 2019, pp. 96--100.

\bibitem{article}
A.~Dyachkov and V.~Rykov, ``Survey of superimposed code theory.''
  \emph{Problems of Control and Information Theory}, vol.~12, pp. 229--242, 01
  1983.

\bibitem{9517965}
J.~Robin and E.~Erkip, ``Capacity bounds and user identification costs in
  {R}ayleigh-fading many-access channel,'' in \emph{2021 IEEE International
  Symposium on Information Theory (ISIT)}, 2021, pp. 2477--2482.

\bibitem{6120391}
C.~L. {Chan}, P.~H. {Che}, S.~{Jaggi}, and V.~{Saligrama}, ``Non-adaptive
  probabilistic group testing with noisy measurements: Near-optimal bounds with
  efficient algorithms,'' in \emph{2011 49th Annual Allerton Conference on
  Communication, Control, and Computing (Allerton)}, Sep. 2011, pp. 1832--1839.

\bibitem{5169989}
D.~Baron, S.~Sarvotham, and R.~G. Baraniuk, ``Bayesian compressive sensing via
  belief propagation,'' \emph{IEEE Transactions on Signal Processing}, vol.~58,
  no.~1, pp. 269--280, 2010.

\bibitem{4085381}
T.~Vercauteren, A.~L. Toledo, and X.~Wang, ``Batch and sequential {B}ayesian
  estimators of the number of active terminals in an {IEEE} 802.11 network,''
  \emph{IEEE Transactions on Signal Processing}, vol.~55, no.~2, pp. 437--450,
  2007.

\bibitem{6477839}
A.~Hazmi, J.~Rinne, and M.~Valkama, ``Feasibility study of {IEEE} 802.11ah
  radio technology for {IoT} and {M2M} use cases,'' in \emph{2012 IEEE Globecom
  Workshops}, 2012, pp. 1687--1692.

\bibitem{6157065}
G.~K. {Atia} and V.~{Saligrama}, ``Boolean compressed sensing and noisy group
  testing,'' \emph{IEEE Transactions on Information Theory}, vol.~58, no.~3,
  pp. 1880--1901, March 2012.

\bibitem{10.5555/1146355}
T.~M. Cover and J.~A. Thomas, \emph{Elements of Information Theory (Wiley
  Series in Telecommunications and Signal Processing)}.\hskip 1em plus 0.5em
  minus 0.4em\relax USA: Wiley-Interscience, 2006.

\bibitem{1057459}
A.~Feinstein, ``A new basic theorem of information theory,'' \emph{Transactions
  of the IRE Professional Group on Information Theory}, vol.~4, 1954.

\bibitem{4797638}
{Jun Luo} and {Dongning Guo}, ``Neighbor discovery in wireless ad hoc networks
  based on group testing,'' in \emph{2008 46th Annual Allerton Conference on
  Communication, Control, and Computing}, 2008, pp. 791--797.

\bibitem{661103}
R.~McEliece, D.~MacKay, and J.-F. Cheng, ``Turbo decoding as an instance of
  pearl's "belief propagation" algorithm,'' \emph{IEEE Journal on Selected
  Areas in Communications}, vol.~16, no.~2, pp. 140--152, 1998.

\bibitem{5394787}
G.~Atia and V.~Saligrama, ``Noisy group testing: An information theoretic
  perspective,'' in \emph{2009 47th Annual Allerton Conference on
  Communication, Control, and Computing (Allerton)}, 2009, pp. 355--362.

\bibitem{6763117}
C.~L. {Chan}, S.~{Jaggi}, V.~{Saligrama}, and S.~{Agnihotri}, ``Non-adaptive
  group testing: Explicit bounds and novel algorithms,'' \emph{IEEE
  Transactions on Information Theory}, vol.~60, no.~5, pp. 3019--3035, 2014.

\bibitem{5707018}
D.~Sejdinovic and O.~Johnson, ``Note on noisy group testing: Asymptotic bounds
  and belief propagation reconstruction,'' in \emph{2010 48th Annual Allerton
  Conference on Communication, Control, and Computing (Allerton)}, 2010, pp.
  998--1003.

\bibitem{hara2022sparse}
Y.~Hara and K.~Kasai, ``Sparse group quantitative pcr testing by belief
  propagation,'' in \emph{2022 IEEE International Symposium on Information
  Theory (ISIT)}.\hskip 1em plus 0.5em minus 0.4em\relax IEEE, 2022, pp.
  2980--2984.

\bibitem{9654225}
V.~K. Amalladinne, A.~K. Pradhan, C.~Rush, J.-F. Chamberland, and K.~R.
  Narayanan, ``Unsourced random access with coded compressed sensing:
  Integrating amp and belief propagation,'' \emph{IEEE Transactions on
  Information Theory}, vol.~68, no.~4, pp. 2384--2409, 2022.

\end{thebibliography}

\end{document}